\documentclass[11pt,english]{article}
\usepackage{color}
\usepackage[T1]{fontenc}
\usepackage[latin9]{inputenc}
\usepackage[a4paper]{geometry}
\usepackage{textcomp}
\usepackage{amsthm}
\usepackage{amsmath}
\usepackage{amssymb}
\usepackage{esint}
\usepackage{bbm}
\usepackage{hyperref}
\usepackage{babel}
\usepackage{xcolor}
\usepackage{cleveref}
\usepackage{authblk}
\usepackage{amscd}
\usepackage{amsfonts}
\usepackage{amsthm}
\usepackage{mathrsfs}
\usepackage{dsfont}
\usepackage{mathtools}
\usepackage{enumerate}
\usepackage{bm}
\usepackage{bbm}
\usepackage{verbatim}
\usepackage{enumitem}
\newlist{abbrv}{itemize}{1}
\setlist[abbrv,1]{label=,labelwidth=1.2in,align=parleft,itemsep=0.1\baselineskip,leftmargin=!}

\makeatletter

\DeclareFontEncoding{LGR}{}{}
\DeclareTextSymbol{\~}{LGR}{126}

\newcommand{\abs}[1]{\left| #1 \right|}
\newcommand{\ups}{\Upsilon_s}
\newcommand{\upt}{\Upsilon_t}

\newcommand{\R}{\mathbb{R}}
\newcommand{\scp}[2]{\big\langle #1 , #2 \big\rangle}
\newcommand{\bra}[1]{\langle #1 |}
\newcommand{\ket}[1]{| #1 \rangle}
\newcommand{\norm}[1]{\left\| #1 \right\| }
\renewcommand{\Re}{\mathrm{Re}}
\renewcommand{\Im}{\mathrm{Im}}
\newcommand{\op}{\mathrm{op}}
\newcommand{\CL}{T}

\usepackage{lipsum}
 
\makeatletter
\newcommand\ackname{Acknowledgements}
\if@titlepage
   \newenvironment{acknowledgements}{%
       \titlepage
       \null\vfil
       \@beginparpenalty\@lowpenalty
       \begin{center}%
         \bfseries \ackname
         \@endparpenalty\@M
       \end{center}}%
      {\par\vfil\null\endtitlepage}
\else
   
\fi
\makeatother

\newcommand{\be}{\begin{equation}}
\newcommand{\ee}{\end{equation}}

\newtheorem{theorem}{Theorem}[section]
\newtheorem{lemma}[theorem]{Lemma}

\newtheorem{corollary}[theorem]{Corollary}

\theoremstyle{definition}
\newtheorem{definition}[theorem]{Definition}
\newtheorem{assumption}[theorem]{Assumption}
\newtheorem{remark}[theorem]{Remark}
\theoremstyle{plain}

\numberwithin{equation}{section}

\allowdisplaybreaks[1]

\title{Landau--Pekar equations and quantum fluctuations for the dynamics of a strongly coupled polaron}
\author[1]{Nikolai Leopold}
\author[2]{David Mitrouskas}
\author[3]{Simone Rademacher} 
\author[4]{Benjamin Schlein}
\author[5]{Robert Seiringer}
\affil[1]{University of Basel, Department of Mathematics and Computer Science, Spiegelgasse 1, 4051 Basel, Switzerland, nikolai.leopold@unibas.ch}
\affil[2]{Universit\"at Stuttgart, Fachbereich Mathematik, Pfaffenwaldring 57,70569 Stuttgart, Germany, mitrouskas@mathematik.uni-stuttgart.de}
\affil[3,5]{Institute of Science and Technology Austria, Am Campus 1, 3400 Klosterneuburg, Austria, simone.rademacher@ist.ac.at, robert.seiringer@ist.ac.at}
\affil[4]{Institute of Mathematics, University of Zurich, Winterthurerstrasse 190, 8057 Zurich, Switzerland, benjamin.schlein@math.uzh.ch}

\begin{document}

\maketitle

\frenchspacing

\begin{abstract}
We consider the Fr\"ohlich Hamiltonian with large coupling constant $\alpha$. For initial data of Pekar product form 
with coherent phonon field and with the electron minimizing the corresponding energy, we provide a norm approximation 
of the evolution, valid up to times of order $\alpha^2$. The approximation is given in terms of a Pekar product state, evolved 
through the Landau--Pekar equations, corrected by a Bogoliubov dynamics taking quantum fluctuations into account. This allows us to show that the Landau-Pekar equations approximately describe the evolution of the electron- and one-phonon reduced density matrices under the Fr\"ohlich dynamics up to times of order $\alpha^2$. \\
\textbf{Keywords:} polaron dynamics, Landau-Pekar equations, quantum fluctuations, Bogoliubov dynamics. \\
\textbf{Mathematics Subject Classification:}  35Q40, 	46N50. 

\end{abstract}

\section{Introduction}

We are interested in the evolution of an electron in an ionic crystal. The electric charge of the electron creates a polarization field in the crystal, which acts back on the electron and modifies its physical properties. In situations in which the extension of the electron is much larger than the lattice spacing, the system can be described by the Fr\"ohlich model \cite{Fr}, which treats the crystal as a continuous medium and describes the polarization of the lattice as excitations (called phonons) of a quantum field. For a review on the current status of some of the mathematical results concerning the Fr\"ohlich polaron, we refer to \cite{FLST,JSM,S}.

We consider the dynamics of the Fr\"ohlich polaron  in the strong coupling limit $\alpha \gg 1$. In particular, we study the evolution of product initial data, describing a coherent phonon field and an electron minimizing the corresponding energy. Approximating the phonons through a classical field, we are led to a system of two coupled nonlinear partial differential equations, known as the Landau--Pekar equations \cite{LaPe}. The coupling parameter $\alpha$ enters the Landau--Pekar equations and makes the electron move much faster than the phonon field, producing a separation of scales, often referred to as adiabatic decoupling \cite{T};  while the electron wave function changes on a time-scale of order one, non-trivial variations of the phonon field only happen over times of order $\alpha^2$. 
Results on adiabatic theorems of the Landau--Pekar equations in one and three spatial dimensions can be found in \cite{FG2} and \cite{LRSS}, respectively.

The goal of this paper is to provide a norm approximation to the microscopic dynamics, valid up to times of order $\alpha^2$, allowing therefore for a non-trivial variation of the phonon field. To reach this goal, the classical evolution predicted by the Landau--Pekar equations has to be modified, taking into account quantum fluctuations that can be described by a time-dependent family of Bogoliubov transformations.  As a corollary of the norm approximation of the many-body dynamics, 
we also prove that the classical Landau--Pekar equations remain valid, up to times of order $\alpha^2$, if we only look at the time-evolution of the electron- and  one-phonon reduced density matrices, without the need of the Bogoliubov modification.  Previous results \cite{FS, FG, G, LRSS} justified the use of the Landau--Pekar equations at most for times small compared with $\alpha^2$, excluding therefore substantial changes of the phonon field. Recently in \cite{M}  a norm-approximation of the dynamics was obtained up to times of order $\alpha^2$, but only for initial data minimizing the Pekar energy functional, leading to a stationary solution of the Landau--Pekar equations. In the following, we will consider a larger class of initial data, given by the product of a general coherent phonon field $\varphi$ with an electron wave function minimizing the energy associated with the field $\varphi$. In particular, this produces non-trivial solutions of the Landau--Pekar equations. An important ingredient of our analysis is the adiabatic theorem proved in \cite{LRSS}. 
Since this theorem requires a gap in the spectrum of the electron Hamiltonian, our results are restricted to times $\vert t \vert \le  \CL \alpha^2$, where $\CL >0$ is a suitably chosen $\alpha$-independent constant, for which the existence of such a gap can be proved. If the existence of a spectral gap of order one were known for longer times, our results would hold for all times of order $\alpha^2$.

\subsection{Model and results}

We consider the Fr\"ohlich model, consisting of an electron with corresponding Hilbert space $L^2 (\mathbb{R}^3)$ coupled to a phonon field, described via the bosonic Fock space $\mathcal{F} = \bigoplus_{n\geq 0} L^2(\mathbb{R}^3)^{\otimes_s^n}$, where the subscript $s$ indicates symmetry under the interchange of variables. The Hilbert space 
of the full system is $\mathcal{H} = L^2(\mathbb{R}^3) \otimes \mathcal{F}$. To study the limit of large coupling $\alpha \gg 1$, it is useful to switch to strong coupling units and to introduce creation and annihilation operators satisfying the canonical commutation relations (CCR)
\begin{align}
\label{eq: commutation relations}
\left[ a_k, a^*_{k'} \right] = \alpha^{-2}  \delta(k - k') \, , \, 
\left[ a_k, a_{k'} \right] = \left[ a^*_{k}, a^*_{k'} \right] = 0
\end{align}
for all $k, k' \in \mathbb{R}^3$. The Fr\"ohlich Hamiltonian then takes the form 
\begin{align}
\label{def:ham}
H = - \Delta + \mathcal{N} +  \phi(G_x) ,
\end{align}
where the Laplacian is acting on the electron, $\mathcal{N}= \int dk\, a_k^* a_k $ denotes the number operator (which equals $\alpha^{-2}$ times the number of phonons) on $\mathcal F$ and
\begin{align}
\label{def:G}
\phi(G_x) =\int d k \, \left( G_x(k) a_k^* + \overline{G_x(k)} a_k \right), \quad G_x (k) =\frac{1}{\abs{k}} e^{- i k \cdot x} .
\end{align}
This Hamiltonian is obtained from the standard form of the Fr\"ohlich Hamiltonian by a suitable change of variables (see, e.g., \cite[Appendix A]{FS}).

In the limit of large $\alpha$, the CCR  (\ref{eq: commutation relations}) suggest that the quantized radiation field approaches 
a classical limit, and to approximate the full evolution generated by the Hamiltonian (\ref{def:ham}) by the corresponding classical Landau--Pekar equations 
\begin{align}
\label{eq:LP}
\begin{cases}
i \partial_t \psi_t = h_{\varphi_t} \psi_t , \\[1mm]
i \alpha^2 \partial_t \varphi_t = \varphi_t + \sigma_{\psi_t} 
\end{cases}
\end{align}
for the electron wave function $\psi_t \in H^1(\R^3)$ and  the classical field $\varphi_t \in L^2(\R^3)$. Here $h_{\varphi} = - \Delta + V_\varphi$ 
and\footnote{The Fourier transform  $\widehat{\cdot}$ is defined for $f \in L^1(\mathbb{R}^3)$ through $\widehat{f}(k) = (2 \pi)^{-3/2} \int d k \, e^{-ik \cdot x} f(x)$.}
\begin{align}
\label{def:pot,sigma}
V_\varphi  (x) = \int \frac{dk }{|k|} \left[ e^{ik \cdot x} \varphi (k) + e^{-ik \cdot x} \overline{\varphi} (k) \right] , \quad \sigma_{\psi}  (k) = (2 \pi)^{3/2} \frac{1}{|k|} \widehat{\vert \psi \vert^2 } (k) .
\end{align}
The well-posedness of the Landau--Pekar equations in the energy space $H^1( \mathbb{R}^3) \times L^2( \mathbb{R}^3)$ 
is shown in \cite[Lemma 2.1]{FG}.

Let us remark that the strong coupling limit is not only a semi-classical limit of the quantum field but corresponds also to an adiabatic limit, which refers to the separation of time-scales in \eqref{eq:LP} as $\alpha \to \infty$. This causes many additional obstacles in the analysis of the strong coupling limit and is also the reason for the particular form of the quantum fluctuations that are introduced below (see Definition \eqref{def:bogo}).

We are interested in the time-evolution generated by the Fr\"ohlich Hamiltonian \eqref{def:ham}, for initial data of Pekar product form 
\begin{align}
\label{eq:initial_state}
\psi_{0} \otimes W ( \alpha^2 \varphi_0 ) \Omega,
\end{align}
where $W( \alpha^2 \varphi_0 )$ is the Weyl operator defined by
\begin{align}
W( f) = e^{a^* ( f) - a(f)}, \quad f \in L^2(\mathbb R^3)
\end{align} 
and  $\Omega$ is the vacuum in the Fock space (but our results also apply more generally, to states in $\mathcal{F}$ having only few phonos). 
We will assume that the initial electron wave function $\psi_0$ is a ground state  of the Schr\"odinger operator $h_{\varphi_0}$ associated with the initial field $\varphi_0$,  introduced in (\ref{eq:LP}). For this reason, we will need the following assumption on $\varphi_0$.   
\begin{assumption}
\label{assumptions}
Let $\varphi_0 \in L^2(\mathbb{R}^3)$ such that
\begin{align}
e(\varphi_0) \coloneqq  \inf \lbrace\langle \psi, h_{\varphi_0} \psi \rangle : \psi \in H^1( \mathbb{R}^3), \| \psi \|_{L^2(\mathbb{R}^3)} =1 \rbrace <0 .
\end{align}
\end{assumption}
This assumption guarantees the existence of a unique positive ground state $\psi_{\varphi_0}$ of $h_{\varphi_0}$ with eigenvalue $e( \varphi_0)$ separated from the rest of the spectrum by a spectral gap 
\begin{align}
\Lambda_0 := \inf_{\substack{\lambda \in \mathrm{spec}( h_{\varphi_0})\\ \lambda \not= e(\varphi_0)}} \vert e(\varphi_0 ) - \lambda \vert > 0.
\end{align} 
Let now $(\psi_t,\varphi_t)$ be the solution of \eqref{eq:LP} with initial data $(\psi_{\varphi_0}, \varphi_0)$. 
As shown in Lemma \ref{lemma:resolvent}, there exists a constant $\CL > 0$ such that, for all times $|t| \leq \CL \alpha^2$, the operator $h_{\varphi_t}$ continues to have a unique positive ground state $\psi_{\varphi_t}$, with eigenvalue $e (\varphi_t) < 0$, separated from the rest of its spectrum by a gap $\Lambda_t$ of order one, independent of $\alpha$. 

The Landau--Pekar equations define an approximation of the evolution of \eqref{eq:initial_state} through product states having the same form as \eqref{eq:initial_state}, with $(\psi_0, \varphi_0)$  replaced by the solution $(\psi_t, \varphi_t)$ of \eqref{eq:LP}. It turns out, however, that in order to obtain a norm-approximation we have to modify this ansatz, implementing non-trivial correlations among phonons. This is achieved via a time-dependent family of Bogoliubov transformations. 

\begin{definition} 
\label{def:bogo}
Let $\varphi_0 $ satisfy Assumption \ref{assumptions} and let $( \psi_t, \varphi_t) \in H^1( \mathbb{R}^3) \times L^2( \mathbb{R}^3)$ denote the solution of the Landau--Pekar equations \eqref{eq:LP} with initial data  $( \psi_{\varphi_0}, \varphi_0 )$. Let $\Upsilon \in \mathcal{F}$ be in the quadratic form domain $\mathcal{Q}(\mathcal{N})$ with $\norm{\Upsilon}_{\mathcal{F}} = 1$. For $|t| \leq \CL \alpha^2$, we define the Bogoliubov dynamics $\upt$ as the solution to
\begin{align}
\label{eq: Bogoliubov dynamics}
\begin{cases}
i \partial_t \upt &= \left( \mathcal{N} - \mathcal{A}_t \right) \upt, \\[1mm]
\Upsilon_0 &= \Upsilon,
\end{cases}
\end{align}
where $\mathcal{A}_t$ is the quadratic operator on $\mathcal{F}$ 
\begin{equation}\label{eq:At}
\begin{split}
\mathcal{A}_t &= \langle \psi_{\varphi_t}, \, \phi (G_{\,\cdot\,}) \,  R_t \,  \phi ( G_{\,\cdot \,} ) \psi_{\varphi_t} \rangle_{L^2(\mathbb{R}^3)}  \\ &= \int \frac{dk}{|k|} \frac{dk'}{|k'|} \langle \psi_{\varphi_t} , e^{-ik \,\cdot \, } R_t e^{ik' \,\cdot \, } \psi_{\varphi_t} \rangle_{L^2(\mathbb{R}^3)}   \left( a_k^* + a_{-k} \right) \left( a_{-k'}^* +   a_{k'} \right)  .
\end{split}\end{equation}
Here, $R_t = q_t \left( h_{\varphi_t} - e( \varphi_t ) \right)^{-1} q_t$ with $q_t = 1- p_t= 1- \vert \psi_{\varphi_t} \rangle \langle \psi_{\varphi_t} \vert $. 
\end{definition}
It follows from Lemma  \ref{lemma:bogo} below that, for all $|t| \leq \CL \alpha^2$, Eq. (\ref{eq:At}) defines $\mathcal{A}_t$ as a self-adjoint operator on the domain of the number  operator $\mathcal{N}$. The well-posedness of the Bogoliubov dynamics \eqref{eq: Bogoliubov dynamics} is shown in Lemma \ref{lemma:bogo2}.

We are now ready to state our main results.

\begin{theorem}\label{thm:main}
Let $\varphi_0$ satisfy Assumption \ref{assumptions} and let $( \psi_t, \varphi_t) \in H^1( \mathbb{R}^3) \times L^2( \mathbb{R}^3)$ denote the solution of the Landau--Pekar equations \eqref{eq:LP} with initial data  $( \psi_{\varphi_0}, \varphi_0 )$. Let $\Upsilon \in \mathcal{F}$ satisfy
$\langle \Upsilon, \mathcal{N}^5 \Upsilon \rangle_{\mathcal{F}} \leq c \alpha^{-10}$ 
for a constant $c>0$, and let $\upt$ denote the solution to \eqref{eq: Bogoliubov dynamics} with initial data $\Upsilon_0=\Upsilon$. Moreover, let 
\begin{equation}
\omega (s) = \alpha^2 \Im \langle \varphi_s,  \partial_s \varphi_s \rangle_{L^2(\mathbb{R}^3)} + \| \varphi_s \|_{L^2(\mathbb{R}^3)}^2.
\end{equation}
Then, there exist $C,\CL>0$ such that 
\begin{equation}\label{eq:thm1}
\left\| e^{-iHt}  \left( \psi_{\varphi_0} \otimes W( \alpha^2 \varphi_0 ) \Upsilon \right)  - e^{-i \int_0^t ds \, \omega (s)} \psi_t  \otimes W( \alpha^2 \varphi_t ) \upt \right\|_{\mathcal{H}} \leq C \alpha^{-1}   
\end{equation}
for all $|t| \leq \CL \alpha^2$. 
\end{theorem}

\begin{remark}
For the proof of the theorem, the persistence of a spectral gap 
\begin{equation}\label{sp-gap}
\Lambda_t := \inf_{\substack{\lambda \in \mathrm{spec}( h_{\varphi_t})\\ \lambda \not= e(\varphi_t)}} \vert e(\varphi_t ) - \lambda \vert
\end{equation} 
of order one is crucial. For this reason, our result \eqref{eq:thm1} is restricted to times $|t| \leq T \alpha^2 $ for which such a gap can be proven, with a constant $T$ depending on the initial field  $\varphi_0$. In fact, let $\widetilde{T} >0$ and assume that there exists  a $\Lambda>0$ such that  $\Lambda_t > \Lambda$ for all $|t| \leq \alpha^2 \widetilde{T}$. Then our proof shows that 
\begin{equation}
\label{eq:thm1new}
\left\| e^{-iHt}  \left( \psi_{\varphi_0} \otimes W( \alpha^2 \varphi_0 ) \Upsilon \right)  - e^{-i \int_0^t ds \, \omega (s)} \psi_t  \otimes W( \alpha^2 \varphi_t ) \upt \right\|_{\mathcal{H}} \leq C \alpha^{-1}   e^{ |t| / \alpha^2} 
\end{equation}
for all $|t| \leq \alpha^2 \widetilde{T}$. 
\end{remark}

With the aid of Theorem \ref{thm:main}, we are able to obtain an approximation for the one-particle reduced density matrices of the electron resp. the phonons in terms of the solution of the Landau--Pekar equations up to times of order $\alpha^2$. The next statement provides a rigorous justification of the time-dependent Landau--Pekar equations starting from the microscopic dynamics generated by the Fr\"ohlich Hamiltonian in the strong coupling limit.

\begin{theorem} \label{corollary:reduced densities} Under the same assumptions as in Theorem \ref{thm:main}, let 
$\Psi_0 = \psi_{\varphi_0} \otimes W( \alpha^2 \varphi_0 ) \Upsilon$, and define the electron reduced density matrix
\begin{equation}
 \gamma^\textnormal{el}_t = \textnormal{Tr}_{\mathcal{F}} \ket{e^{- i Ht} \Psi_0} \bra{e^{- i Ht} \Psi_0} . 
\end{equation}
Then there exist constants $C,\CL > 0$ such that 
\begin{align}
\label{eq: bdelectron density}
\left\| \gamma^\textnormal{el}_t - |\psi_t \rangle \langle \psi_t| \right\|_\textnormal{tr} \leq C \alpha^{-1} 
\end{align}
for all $|t| \leq \CL \alpha^2$, where $\| . \|_\textnormal{tr}$ denotes the trace norm. 

If we additionally assume that $\varphi_0 \in L^2 (\mathbb{R}^3, |k|^{1/2} dk)$, we also find that for all $|t| \leq \CL \alpha^2$
\begin{align}
\label{eq: bdphonon density}
 \left\| \gamma^\textnormal{ph}_t  - |\varphi_t \rangle \langle \varphi_t | \right\|_\textnormal{tr} \leq  C \big( \alpha^{-1/4} +\alpha^{-2}\big) ,
\end{align}
where $\gamma^\textnormal{ph}_t$ is the one-phonon reduced density matrix defined through its integral kernel 
\begin{equation}
\gamma^\textnormal{ph}_t (k , k') = \langle e^{-i Ht} \Psi_0 , a_{k'}^* a_{k} e^{-i Ht} \Psi_0 \rangle_{\mathcal H} 
\end{equation}
for $k,k' \in \mathbb{R}^3$.
\end{theorem}

\begin{remark} While \eqref{eq: bdelectron density} is an almost immediate consequence of the norm approximation \eqref{eq:thm1}, the derivation of the bound for the phonons \eqref{eq: bdphonon density} is more elaborate. For its proof, we first derive an estimate for the number of phonons outside the coherent state, of the form
\begin{equation}\label{eq:bdN} 
\| \mathcal{N}^{1/2} W^* (\alpha^2 \varphi_t) e^{-i Ht} \Psi_0 \|_\mathcal{H}^2 \leq C  \big( \alpha^{-1/2} + \alpha^{-2} \big) ,
\end{equation}
which is then used to obtain \eqref{eq: bdphonon density}. Let us remark that without the additional regularity assumption $\varphi_0 \in L^2 (\mathbb{R}^3, |k|^{1/2} dk)$, we would obtain only a bound for the expectation value of $\mathcal N^{1/2}$, namely\footnote{The bound (\ref{eq:bdN4}) follows by proceeding as in (\ref{eq:N12-1}), (\ref{eq:N12-2}), with $\mathcal{N}_\leq$ replaced by $\mathcal{N}^{1/2}$, and applying Lemma \ref{lemma:N}.}
\begin{equation}\label{eq:bdN4} 
\| \mathcal{N}^{1/4} W^* (\alpha^2 \varphi_t) e^{-i Ht} \Psi_0 \|_\mathcal{H}^2 \leq C  \big(\alpha^{-1} + \alpha^{-2} \big).
\end{equation}
\end{remark}

While Theorem \ref{corollary:reduced densities} shows that the Landau--Pekar equations provide a good approximation for the one-particle reduced densities associated with the dynamics generated by the Fr\"ohlich Hamiltonian,  the introduction of the Bogoliubov dynamics (\ref{eq: Bogoliubov dynamics}) to capture quantum fluctuations is  crucial  to obtain a  norm-approximation of the full wave function in $\mathcal{H}$ for times of order $\alpha^2$. 
This is quantified in the following remark (whose proof is postponed to Section \ref{rmk:bogo}).

\begin{remark}
\label{rmk:B}
Under the same assumptions as in Theorem \ref{thm:main}, with $\Upsilon = \Omega$, and for $\delta >0$ sufficiently small, there exists a constant $C_\delta >0 $ such that for $ t = \delta \alpha^2$
\begin{align}
\| e^{-iHt}  \left( \psi_{\varphi_0} \otimes W( \alpha^2 \varphi_0 ) \Omega \right)  - e^{-i \int_0^t ds \, \omega (s)} \psi_t  \otimes W( \alpha^2 \varphi_t )  \Omega \|_{\mathcal{H}}^2  &\geq C_\delta 
\end{align}
for large $\alpha$. 
\end{remark}

The rest of the article is organized as follows. 
In Section \ref{section: preliminaries} we introduce some relevant notation, recall known properties of the Landau--Pekar equations and prove well-posedness of the Bogoliubov dynamics \eqref{eq: Bogoliubov dynamics} together with helpful bounds involving the operator $\mathcal A_t$. Theorem \ref{thm:main}, Theorem \ref{corollary:reduced densities}  and Remark \ref{rmk:B} are proven in Sections \ref{subsection: proof main theorem}, \ref{subsection: proof corollary reduced densities} and \ref{subsection: proof of remark bogo}, respectively.
 

\paragraph{Comparison with the literature.} 
The mathematically rigorous derivation of the Landau--Pekar equations from the Fr\"ohlich model in the strong coupling limit was initiated in \cite{FS}, where product states with stationary phonon field have been used to approximate the evolution of Pekar product states. Taking into account the evolution of the phonons, this result was improved in \cite{FG}, where $\psi_t \otimes W(\alpha^2 \varphi_t) \Omega$, with $(\psi_t,\varphi_t)$ solving the Landau--Pekar equations, was proven to approximate the many-body evolution up to times $|t| \ll \alpha$.

For the minimizer $(\psi^{\rm P}, \varphi^{\rm P})$ of the Pekar functional (which is a stationary solution of \eqref{eq:LP}, up to a phase) 
as particular initial state the validity of the Landau--Pekar equations was proven for times $t\ll \alpha^2$ in \cite{G}. An important observation in \cite{G}, which also plays an essential role in this work, is that the separation of time scales in the Landau--Pekar equations and the spectral gap of $h_{\varphi^{\rm P}}$ between its lowest eigenvalue and the rest of the spectrum give rise to an oscillatory phase,  which effectively keeps the electron  from leaving the ground state. Later, in \cite{LRSS} an adiabatic theorem for the Landau--Pekar equations was proven and used to show their accuracy for $t\ll \alpha^2$ for initial states of the form \eqref{eq:initial_state} with $\varphi_0$ satisfying Assumption \ref{assumptions} and $\psi_0 = \psi_{\varphi_0}$. 

The importance of the adiabatic theorem for the derivation of the Landau--Pekar equations was already realized in \cite{F} and an adiabatic theorem in one spatial dimension was proved in \cite{FG2}. The stationary case was revisited in \cite{M} where the norm approximation \eqref{eq:thm1} was proven for the particular initial state $\psi^{\rm P}\otimes W(\alpha^2 \varphi^{\rm P})\Upsilon$ for all times of order $\alpha^2$.

We note that the Landau--Pekar equations can also be derived in a many-body mean-field limit \cite{LMS}, where there is no separation of time scales, however. Similar results for related models are obtained in \cite{Fa, LPe,AF,LP,GNV, CCFO, CFO}.

\section{Preliminaries}

\label{section: preliminaries}

\subsection{Notation}
In the following, the letter $C$ is used as a generic constant independent of $t$ and $\alpha$. The $L^p$-norm of a function $f \in L^p(\mathbb{R}^3, \mathbb{C})$ with $0 < p \leq \infty$ is denoted by $\norm{f}_p$. We use
$\| \cdot \|$ and $\langle \cdot, \cdot \rangle$ for the norm and inner product on $\mathcal{H}= L^2(\mathbb{R}^3) \otimes \mathcal{F}$. Norms and inner products on different Hilbert spaces will always be indicated with the corresponding subscript.
To simplify the notation, we define for $f \in L^2(\mathbb{R}^3)$ the creation operator $a^*(f)$, the annihilation operator $a(f)$ and the field operator $\phi(f)$ by
\begin{align}
a^*(f) = \int d k \, f(k) a_k^* \, , \, \,
a(f) = \int d k \, \overline{f(k)} a_k
\, , \, \,
\phi(f) = a(f) + a^*(f) .
\end{align}
They are bounded with respect to the number  operator, i.e.,
\begin{align}
\label{eq: bounds for the creation and annihilation operators}
\| a(f) \xi \|_{\mathcal{F}} &\leq \| f \|_2 \| \mathcal{N}^{1/2} \xi \|_{\mathcal{F}}
\, , \,
\| a^*(f) \xi \|_{\mathcal{F}} \leq \| f \|_2 \| \left( \mathcal{N}+ \alpha^{-2} \right)^{1/2} \xi \|_{\mathcal{F}}
\quad \text{for all} \,
\xi \in \mathcal{F}.
\end{align}

\subsection{Properties of the Landau--Pekar equations}

In this section,  we collect useful properties of the Landau--Pekar equations. Their well-posed\-ness is shown in the following lemma.

\begin{lemma}[\cite{FG}, Lemma 2.1]
\label{lemma:LP}
For any $ ( \psi_0, \varphi_0) \in H^1( \mathbb{R}^3) \times L^2( \mathbb{R}^3)$, there is a unique global solution $( \psi_t, \varphi_t)$ of the Landau--Pekar equations \eqref{eq:LP}. The $L^2$-norm of the electron wave function is conserved,
$
\| \psi_t \|_{2} = \| \psi_0 \|_{2}  
$
for all $t \in \mathbb{R}^3$, and there exists a constant $C$ such that 
\begin{align}
\| \psi_t \|_{H^1( \mathbb{R}^3)} \leq C, \hspace{0.3cm} \| \varphi_t \|_{2} \leq C 
\end{align}
for all $\alpha >0$ and all $t \in \mathbb{R}$.
\end{lemma}

The following lemma shows properties of the potential $V_\varphi$ and $\sigma_\psi$ defined in \eqref{def:pot,sigma} (see also \cite[Lemma III.2]{LRSS}).

\begin{lemma}\label{lemma:Potental} 
For  $V_\varphi$ defined in \eqref{def:pot,sigma}, there exists a constant $C>0$ such that for every $ \psi \in H^1( \mathbb{R}^3)$ and $ \varphi \in L^2( \mathbb{R}^3)$ 
\begin{align}
\| V_\varphi \|_6 \leq C  \| \varphi \|_2
 \hspace{0.3cm} 
 and \hspace{0.3cm} 
\| V_\varphi \psi \|_2 \leq  C  \| \varphi \|_2 \, \|\psi \|_{H^1( \mathbb{R}^3)} .
\end{align}
Moreover, let $\sigma_\psi$ be defined  in \eqref{def:pot,sigma}.  Then, there exists $C>0$ such that for all $\psi_1, \psi_2 \in H^1( \mathbb{R}^3) $
\begin{align}
\label{eq: bound sigma-psi}
\| \sigma_{\psi_1} \|_{2} \leq C \| \psi_1 \|_{H^1( \mathbb{R}^3)}^2 , \quad \| \sigma_{\psi_1} - \sigma_{\psi_2} \|_2 \leq C \left( \| \psi_1 \|_{H^1( \mathbb{R}^3)} + \| \psi_2 \|_{H^1( \mathbb{R}^3)} \right) \| \psi_1 - \psi_2 \|_2 .
\end{align}
\end{lemma}

\begin{proof}
The first three inequalities follow from \cite[Lemma III.2]{LRSS}. For the last one, we write
\begin{align}
\sigma_{\psi_1}(k) - \sigma_{\psi_2}(k) & =  \frac{1}{|k|} \left( \langle \psi_1,\,  e^{- ik \cdot} \psi_1 \rangle_{L^2(\mathbb{R}^3)} - \langle \psi_2, \, e^{- ik \cdot} \psi_2 \rangle_{L^2(\mathbb{R}^3)} \right) \notag\\
& =  \frac{1}{|k|} \left( \langle \psi_1 - \psi_2, \, e^{-ik \cdot} \psi_1 \rangle_{L^2(\mathbb{R}^3)} + \langle \psi_2, e^{-ik \cdot} \left( \psi_1 - \psi_2 \right)\rangle_{L^2(\mathbb{R}^3)} \right) .
\end{align}
Thus, 
\begin{align}
\label{eq:sigma_estimate}
\| \sigma_{\psi_1} - \sigma_{\psi_2} \|_2^2 & \leq 2 \int \frac{dk}{|k|^2} \left(\vert \langle \psi_1 - \psi_2, \, e^{-ik \cdot} \psi_1 \rangle_{L^2(\mathbb{R}^3)} \vert^2 + \vert\langle \psi_2, e^{-ik \cdot} \left( \psi_1 - \psi_2 \right)\rangle_{L^2(\mathbb{R}^3)} \vert^2 \right) .
\end{align}
The Hardy--Littlewood--Sobolev and the Sobolev inequalities imply
\begin{align}
\int \frac{dk}{|k|^2} \vert \langle \psi_1 - \psi_2, \, e^{-ik \cdot} \psi_1 \rangle_{L^2(\mathbb{R}^3)} \vert^2& = C\int \frac {dxdy}{| x-y|}  \left( \psi_1 - \psi_2 \right) (x) \, \overline{\left( \psi_1 - \psi_2 \right) (y) } \, \overline{\psi_1 (x)} \, \psi_1 (y)
\notag \\
& \leq C \| \psi_1 \, \overline{\left( \psi_1 - \psi_2 \right) } \|_{6/5}^2\notag\\
& \leq C \| \psi_1 \|_3^2 \| \psi_1 - \psi_2 \|_2^2 \notag\\
& \leq  C \| \psi_1 \|^2_{H^1( \mathbb{R}^3)} \, \| \psi_1 - \psi_2 \|_2^2 .
\end{align}
The second term on the r.h.s. of \eqref{eq:sigma_estimate} can be bounded in the same way. Hence, the second inequality of \eqref{eq: bound sigma-psi} follows. 
\end{proof}
 
For the effective dynamics, the ground state $\psi_{\varphi_t}$ of the operator $h_{\varphi_t}$ plays an important role. The following lemma concerning its time evolution is proven in \cite{LRSS}.

\begin{lemma}[\cite{LRSS}, Lemma IV.1 and Remark III.1]
\label{lemma:minimizer}
Let $\varphi_0$ satisfy Assumption \ref{assumptions}. Then, there exists a unique positive and normalized ground state $\psi_{\varphi_0}$ of $h_{\varphi_0} = - \Delta + V_{\varphi_0}$. 
Moreover, let $( \psi_t, \varphi_t) \in H^1( \mathbb{R}^3) \times L^2( \mathbb{R}^3)$ denote the solution of the Landau--Pekar equations \eqref{eq:LP} with initial data  $( \psi_{\varphi_0}, \varphi_0 )$. There exists a constant $T >0$ such that for all $|t|  \leq T  \alpha^2$ the following properties hold: There exists a unique positive and normalized ground state $\psi_{\varphi_t}$ of $h_{\varphi_t} = - \Delta + V_{\varphi_t}$ with corresponding eigenvalue $e( \varphi_t) < 0$. It satisfies 
\begin{align}\label{partpsi}
\partial_t \psi_{\varphi_t} = \alpha^{-2} R_t  V_{i \varphi_t} \psi_{\varphi_t}
 \hspace{0.3cm} \mathrm{with} \hspace{0.3cm} R_t = q_t (h_{\varphi_t} - e(\varphi_t))^{-1} q_t ,
\end{align}
where $q_{t} = 1- | \psi_{\varphi_t} \rangle \langle \psi_{\varphi_t} |$ denotes the projection onto the subspace of $L^2 ( \mathbb{R}^3)$ orthogonal to the span of $\psi_{\varphi_t}$. Moreover, there exists $C>0$ such that
\begin{align}
\| \psi_{\varphi_t} \|_{H^1( \mathbb{R}^3)} \leq C.
\end{align}
\end{lemma}

One can also show that under the assumptions of Lemma~\ref{lemma:minimizer} there exists a constant $T>0$  such that the spectral gap $\Lambda_t$ of the Hamiltonian $h_{\varphi_t}$, defined in \eqref{sp-gap}, remains of order one for all times $|t| \leq \CL \alpha^2$. In particular, this leads to bounds  on the resolvent $R_t $ and its time derivative, uniformly in $\alpha$. We summarize the relevant statements in the following lemma (compare with \cite[Lemmas II.1 and IV.2]{LRSS}). 

\begin{lemma}
\label{lemma:resolvent}
Let $\varphi_0$ satisfy Assumption \ref{assumptions} and let  $( \psi_t, \varphi_t) \in H^1( \mathbb{R}^3) \times L^2( \mathbb{R}^3)$ denote the solution of the Landau--Pekar equations \eqref{eq:LP} with initial data  $( \psi_{\varphi_0}, \varphi_0 )$. Then, for all $\Lambda$ with $0 < \Lambda < \Lambda_0$ there exists $\CL >0$ such that 
\begin{align}
\Lambda_t \geq \Lambda \quad \mathrm{for \, all } \quad \vert t \vert \leq \CL \alpha^2 .
\end{align}
With $p_t =  | \psi_{\varphi_t} \rangle \langle \psi_{\varphi_t} |$ we have\footnote{We use the shorthand notation $\partial_t f (t) = \dot{f} (t)$ for the time derivative.}
\begin{equation}\label{eq: derivative rho}
\alpha^2 \dot{R}_t  =-  p_t V_{i\varphi_t} R_t^2  -  R_t^2 V_{i \varphi_t} p_t +  R_t  \left( V_{i \varphi_t} - \langle \psi_{\varphi_t}, V_{i \varphi_t} \psi_{\varphi_t} \rangle \right) R_t .
\end{equation}
Moreover, the bounds 
\begin{align}\label{eq: bound on R}
\| R_t \|_{\textnormal{op}} \leq \Lambda^{-1},  \hspace{0.5cm} \| ( - \Delta +1)^{1/2} R_t^{1/2} \|_{\textnormal{op}} \leq C ( 1 + \Lambda^{-1})^{1/2} ,
\end{align} 
\begin{align}\label{eq: bound on dot R}
\| \left( - \Delta +1 \right)^{1/2} \dot{R}_t \left( - \Delta +1 \right)^{1/2}\|_{\textnormal{op}} \leq C \alpha^{-2} \left( 1 + \Lambda^{-1} \right)^2
\end{align}
and 
\begin{equation}\label{ps}
\| \partial_t \sigma_{\psi_{\varphi_t}} \|_{2} \leq C \alpha^{-2} \left( 1 + \Lambda^{-1} \right)^{1/2}
\end{equation}
hold for all $|t| \leq \CL \alpha^2$, with $\|\, \cdot\,\|_{\rm op}$ denoting the operator norm, and  $C >0$ depending only on $\varphi_0$. 
\end{lemma}

\begin{proof}
The  identity \eqref{eq: derivative rho} as well as the bounds \eqref{eq: bound on R} follow from \cite[Lemma IV.2]{LRSS}. Inequality \eqref{eq: bound on dot R} follows immediately from  \eqref{eq: derivative rho} and \eqref{eq: bound on R} in combination with Lemmas \ref{lemma:Potental} and \ref{lemma:minimizer}, since the latter imply that 
\begin{equation}\label{cor:zsf_prelim}
\| R_t^{1/2} V_{\varphi_t} \|_{\textnormal{op}} \leq C( 1 + \Lambda^{-1})^{1/2} , \quad \| V_{\varphi_t} \psi_{\varphi_t} \|_2  \leq C \| \psi_{\varphi_t} \|_{H^1( \mathbb{R}^3)} \leq C
\end{equation}
for all $|t| \leq \CL \alpha^2$. 
 To prove the last inequality we note that an application of 
 Lemma \ref{lemma:minimizer} yields 
\begin{equation}
\partial_t \sigma_{\psi_{\varphi_t}}(k)
= 2 \alpha^{-2} \frac{1}{\abs{k}} \,  \Re \langle R_t V_{i \varphi_t} \psi_{\varphi_t}, e^{-ik \, \cdot\, } \psi_{\varphi_t} \rangle_{L^2(\mathbb{R}^3)} .
\end{equation}
We can then proceed as in the proof of Lemma~\ref{lemma:Potental}, using the Hardy--Littlewood--Sobolev inequality as well as \eqref{cor:zsf_prelim}, to arrive at \eqref{ps}.
\end{proof}  

For the proof of Theorem \ref{thm:main} we shall also need the following adiabatic theorem proved in  \cite{LRSS}.

\begin{theorem}[\cite{LRSS}, Theorem II.1]
\label{thm:adiabatic}
Let $\varphi_0 $ satisfy Assumption \ref{assumptions}. Let $( \psi_t, \varphi_t) \in H^1( \mathbb{R}^3) \times L^2( \mathbb{R}^3)$ denote the solution of the Landau--Pekar equations\eqref{eq:LP} with initial data  $( \psi_{\varphi_0}, \varphi_0 )$. Then, there exist $\CL , C >0$ such that
\begin{align}
\| \psi_t - e^{- i \int_0^t du \, e( \varphi_u)} \psi_{\varphi_t} \|_2 \leq C  \alpha^{-2} 
\end{align}
for all $|t| \leq \CL \alpha^2$. 
\end{theorem}

\subsection{Bounds on creation and annihilation operators}

For the proof of Theorem~\ref{thm:main} we shall need the following bounds. 

\begin{lemma}
\label{lemma:G}
Let $G_x(k) = \frac{1}{|k|} e^{- ik \cdot x}$. There exists $C>0$ such that 
\begin{align}\label{eq: bounds lemma:G}
\| a \left( G_{\, \cdot\, } \right) ( - \Delta +1)^{-1/2} \Psi \|  & \leq C \| \mathcal{N}^{1/2} \Psi \| \notag \\
\| ( - \Delta +1 )^{-1/2} a^* \left( G_{\, \cdot \, } \right) \Psi \| & \leq C \| ( \mathcal{N} +\alpha^{-2} )^{1/2} \Psi \| 
\end{align}
for all $\Psi \in L^2( \mathbb{R}^3) \otimes \mathcal{F}$.
\end{lemma}

\begin{proof}
The first inequality follows from  \cite[Lemma~10]{FS}. The second one is an immediate consequence, using the CCR. 
\end{proof}

Note that whenever $\Psi$ is not in the domain of $\mathcal N^{1/2}$, the right side in \eqref{eq: bounds lemma:G} is infinite and thus the inequality holds trivially. We use this convention throughout this section.

As a consequence of Lemmas \ref{lemma:minimizer}, \ref{lemma:resolvent} and \ref{lemma:G} we obtain the following corollary. 

\begin{corollary}
\label{cor:zsf_G}
\textit{Let $G_x (k)= \frac{1}{|k|} e^{ - ik \cdot x}$ and $\Upsilon \in \mathcal{F}$. Moreover, let $\varphi_0$ satisfy Assumption \ref{assumptions} and let $\left( \psi_t, \varphi_t \right) \in H^1( \mathbb{R}^3) \times L^2( \mathbb{R}^3)$ denote the solution of the Landau--Pekar equations with initial data  $( \psi_{\varphi_0}, \varphi_0 ) \in H^1( \mathbb{R}^3) \times L^2( \mathbb{R}^3 )$. Then, there exist $C, \CL >0$ such that 
\begin{equation}
\|  a \left( G_{\, \cdot \, } \right) \psi_{\varphi_t} \otimes \Upsilon \| 
\leq  C \| \psi_{\varphi_t} \|_{H^1( \mathbb{R}^3)}  \| \left( \mathcal{N} + \alpha^{-2}\right)^{1/2} \Upsilon\|_{\mathcal{F}}  \leq  C    \| \left( \mathcal{N} + \alpha^{-2}\right)^{1/2} \Upsilon \|_{\mathcal{F}}
\end{equation}
and
\begin{equation}
\| R_t^{1/2} a^* \left( G_{\, \cdot \,} \right) \psi_{\varphi_t} \otimes \Upsilon \|
\leq C \| \psi_{\varphi_t} \|_{L^2( \mathbb{R}^3)}  \| \left( \mathcal{N} + \alpha^{-2}\right)^{1/2} \Upsilon\|_{\mathcal{F}} 
\leq  C    \| \left( \mathcal{N} + \alpha^{-2}\right)^{1/2} \Upsilon \|_{\mathcal{F}}
\end{equation}
for all $|t| \leq \CL \alpha^2$. }
\end{corollary}


\subsection{Bogoliubov dynamics}

In this section we shall provide  bounds for the operator $\mathcal{A}_t$ that will be useful in the proof, and in particular allow us to prove the well-posedness of the Bogoliubov dynamics \eqref{eq: Bogoliubov dynamics}.
\begin{lemma}
Let $\varphi_0$ satisfy Assumption \ref{assumptions}. Let $( \psi_t, \varphi_t) \in H^1( \mathbb{R}^3) \times L^2( \mathbb{R}^3)$ denote the solution of the Landau--Pekar equations \eqref{eq:LP} with initial data  $( \psi_{\varphi_0}, \varphi_0 )$. For $\mathcal{A}_t$ as in Definition \ref{def:bogo}, there exist $C, \CL >0$ such that 
\label{lemma:bogo}
\begin{align}
\label{eq: norm-bound for A}
\| \mathcal{A}_t  \Psi \|_{\mathcal{F}} &\leq C  \| ( \mathcal{N} + \alpha^{-2} ) \Psi \|_{\mathcal{F}},
\\
\label{eq: norm-bound for commutator of A with N}
\norm{\left[ \mathcal{N}, \mathcal{A}_t \right] \Psi}_{\mathcal{F}} &\leq C \alpha^{-2}  \norm{\left( \mathcal{N} + \alpha^{-2} \right) \Psi}_{\mathcal{F}} ,
\\
\label{eq: norm-bound for double of commutator A with N}
\norm{ \left[ \mathcal{N} ,\left[ \mathcal{N}, \mathcal{A}_t \right] \right] \Psi}_{\mathcal{F}} &\leq C \alpha^{-4}  \norm{\left( \mathcal{N} + \alpha^{-2} \right) \Psi}_{\mathcal{F}},
\\
\label{eq: norm-bound for triple of commutator A with N}
\norm{\left[ \mathcal{N} , \left[ \mathcal{N} ,\left[ \mathcal{N}, \mathcal{A}_t \right] \right] \right] \Psi}_{\mathcal{F}} &\leq C \alpha^{-6}  \norm{\left( \mathcal{N} + \alpha^{-2} \right) \Psi}_{\mathcal{F}},
\\
\label{eq: form bound derivative of A}
\| \dot{\mathcal{A}}_t \Psi \|_{\mathcal{F}}
&\leq C \alpha^{-2}  \norm{\left( \mathcal{N} + \alpha^{-2} \right) \Psi}_{\mathcal{F}} 
\end{align}
for all $| t | \leq \CL \alpha^2$ and $\Psi \in \mathcal{F}$.
\end{lemma}

\begin{proof}
From \eqref{eq:At} and the CCR we can write
\begin{align}
\label{eq: normal-order A}
\mathcal{A}_t = \int dk dl \,  F_t(k,  l) \, \left( a_k^*a_{-l}^* + a_k^* a_{l} +  a^*_{-l} a_{-k} + a_{-k} a_{l} \right) + \alpha^{-2} \int dk \, F_t(k,k)
\end{align}
with
\begin{equation}\label{def:F}
F_t(k,l) =  |k|^{-1} |l|^{-1}\scp{\psi_{\varphi_t}}{e^{ - ik \,\cdot\,} R_t e^{ i l \, \cdot\, } \psi_{\varphi_t}}_{L^2(\mathbb{R}^3)}.
\end{equation}
We have  
\begin{align}
\int dk \,  F_t ( k, k )   &= \int \frac{dk}{\abs{k}^2} \norm{R_t^{1/2} e^{ i k \cdot} \psi_{\varphi_t}}_2^2  \notag\\
&\leq    \norm{R_t^{1/2} \left( 1 - \Delta \right)^{1/2}}_{\op}^2
\int \frac{dk}{\abs{k}^2} \norm{\left( 1 - \Delta \right)^{-1/2} e^{ i k \cdot} \psi_{\varphi_t}}_2^2  
\nonumber \\
&=  \norm{R_t^{1/2} \left( 1 - \Delta \right)^{1/2}}_{\op}^2
\scp{\psi_{\varphi_t}}{\int \frac{dk}{\abs{k}^2} \left( 1 + \abs{ i \nabla + k}^2 \right)^{-1} \psi_{\varphi_t}}_{L^2(\mathbb{R}^3)} .
\end{align}
Since
\begin{align}
\norm{\int \frac{dk}{\abs{k}^2} \left( 1 + \abs{ i \nabla + k}^2 \right)^{-1}}_{\op}
&= \sup_{p \in \mathbb{R}}
\int \frac{dk}{\abs{k}^2 \left( 1 + \abs{p + k}^2 \right)}  < \infty,
\end{align}
we conclude with Lemma \ref{lemma:resolvent} that 
\begin{align}
\label{eq:estimate_F}
\int dk \, \vert F_t ( k, k ) \vert  &\leq C 
\end{align}
for all $| t | \leq \CL \alpha^2$. Similarly, we find 
\begin{align}
\| F_t \|_{L^2( \mathbb{R}^3 \times \mathbb{R}^3)} &\leq \int \frac{dk}{\abs{k}^2} \norm{R_t^{1/2} e^{ i k \cdot} \psi_{\varphi_t}}_2^2  \leq C 
\end{align}
for all $| t | \leq \CL \alpha^2$. Using the bound 
\begin{align}
\norm{\int dk \, dl \, f(k,l) a^{\sharp_1}_k a^{\sharp_2}_l  \Psi}_{\mathcal{F}}
&\leq \sqrt{2} \norm{\left( \mathcal{N} + \alpha^{-2} \right) \Psi}_{\mathcal{F}} \left( \norm{f}_{L^2( \mathbb{R}^3 \times \mathbb{R}^3)} + \int dk \, \abs{f(k,k)} \right)
\end{align}
(which easily follows from the usual estimates of the creation and annihilation operators, see \cite[Lemma~2.1]{BS}) with $\sharp_1, \sharp_2 \in \{\cdot, * \}$ and $\Psi \in \mathcal{F}$, we obtain \eqref{eq: norm-bound for A}.

Since
\begin{align} 
\int dk  \, dl \, F_t(k,- l) \left( a_{-k} a_{-l} - a_k^* a_l^* \right) 
&= - \frac{1}{2} \alpha^2 \left[ \mathcal{N}, \mathcal{A}_t  \right]
= - \frac{1}{8} \alpha^6 \left[ \mathcal{N} , \left[ \mathcal{N} , \left[ \mathcal{N}, \mathcal{A}_t \right] \right] \right] ,
\notag \\
\label{eq: commutator of A with number operator}
 \int dk \, dl \, F_t(k, - l)
\left( a_{-k} a_{-l} + a_k^* a_l^* \right)
&=  \frac{1}{4} \alpha^4  \left[ \mathcal{N} , \left[ \mathcal{N}, \mathcal{A}_t \right] \right]
\end{align}
have a similar structure as $\mathcal{A}_t$, one obtains \eqref{eq: norm-bound for commutator of A with N}--\eqref{eq: norm-bound for triple of commutator A with N} in the same way.
Using Lemmas~\ref{lemma:minimizer} and~\ref{lemma:resolvent}  as well as \eqref{cor:zsf_prelim}
the operator
\begin{equation}
\alpha^2 \dot{\mathcal{A}}_t
= \int dk dl \,   I_t(k,  l)   \left( a_k^*a_{-l}^* + a_k^* a_{l} +  a^*_{-l} a_{-k}
 + a_{-k} a_{l} \right) 
 + \alpha^{-2} \int dk \,  I_t(k, k)   
\end{equation}
where 
\begin{align}
I_t(k,l)
&= \abs{k}^{-1} \abs{l}^{-1} \big(
\scp{\psi_{\varphi_t}}{e^{ - ik \cdot} R_t e^{ i l \cdot} R_t V_{i \varphi_t} \psi_{\varphi_t}}_{L^2(\mathbb{R}^3)}
+  \scp{ R_t V_{i \varphi_t} \psi_{\varphi_t}}{e^{ - ik \cdot} R_t e^{ i l \cdot} \psi_{\varphi_t}}_{L^2(\mathbb{R}^3)} \big) 
\notag \\
& \quad + \alpha^2 \abs{k}^{-1} \abs{l}^{-1}
\scp{\psi_{\varphi_t}}{e^{ - ik \cdot} \dot{R}_t e^{ i l \cdot} \psi_{\varphi_t}}_{L^2(\mathbb{R}^3)}
\end{align}
can be bounded by similar arguments.
\end{proof}

\begin{lemma}
\label{lemma:bogo2}
With the same assumptions as in Lemma~\ref{lemma:bogo}, there exists 
for every $\alpha >0$ and state $\Upsilon$ in the quadratic form domain $\mathcal Q \left( \mathcal{N} \right)$ with $\norm{\Upsilon}_{\mathcal{F}}= 1$ a unique solution of 
\begin{align}
\begin{cases}
i \partial_t \upt &= \left( \mathcal{N} - \mathcal{A}_t \right) \upt  \\
\Upsilon_0 &= \Upsilon
\end{cases}
\end{align}
such that $ \Upsilon_{\cdot} \in C^0 \left([0, \CL \alpha^2), \mathcal{F} \right) \cap L^{\infty} \left( [0, \CL \alpha^2), Q(\mathcal{N}) \right)$ and  $\norm{\upt}_{\mathcal{F}} = 1$. Moreover, for $0\leq j\leq 5$
\begin{equation} \label{eq: bound on the number operator}
\alpha^{j}  \| \left( \mathcal{N} + \alpha^{-2}\right)^{j/2}  \upt \|_{\mathcal{F}} \leq C \alpha^{5}  \| \left( \mathcal{N} + \alpha^{-2}\right)^{5/2}  \Upsilon \|_{\mathcal{F}} 
\end{equation}
holds for all $|t | \leq \CL \alpha^2$. 
\end{lemma}

\begin{proof}
The first claim follows from \cite[Theorem 8]{LNS}, rescaling the time variable and setting  $H(t) = \CL \alpha^2 (\mathcal{N} - \mathcal A_{ \CL \alpha^2 t})$ and $A = B = (1 + \alpha^2) \mathcal{N} + 1 \geq 1$,  utilizing  Lemma \ref{lemma:bogo}.
To show \eqref{eq: bound on the number operator}, we estimate
\begin{align}
\frac{d}{dt} \norm{ \left( \mathcal{N} + \alpha^{-2} \right)^{5/2} \,  \upt}_{\mathcal{F}}^2
&= -i \scp{ \upt}{\left[ \mathcal{N}^5  , \mathcal{A}_t \right] \upt}_{\mathcal{F}}
\nonumber \\
&\leq C \abs{\scp{\mathcal{N}^2  \upt}{ \left[ \mathcal{N}, \mathcal{A}_t \right]  \mathcal{N}^2  \upt}_{\mathcal{F}}}
\nonumber \\
&\quad + 
C \abs{\scp{\mathcal{N}^2  \upt}{\left[ \mathcal{N} , \left[ \mathcal{N}, \mathcal{A}_t \right] \right] \mathcal{N}  \upt}_{\mathcal{F}}}
\nonumber \\
&\quad + C \abs{\scp{\mathcal{N}^2  \upt}{\left[ \mathcal{N} , \left[ \mathcal{N} , \left[ \mathcal{N}, \mathcal{A}_t \right] \right] \right] \upt}_{\mathcal{F}}} 
\nonumber \\
&\leq 
C \alpha^{-2} \norm{\left( \mathcal{N} + \alpha^{-2} \right)^{5/2}   \upt}_{\mathcal{F}}^2 
\end{align}
again with the aid of Lemma \ref{lemma:bogo}.
Gr\"onwall's lemma thus gives
\begin{align}
\norm{\left( \mathcal{N} + \alpha^{-2} \right)^{5/2}  \upt}_{\mathcal{F}}^2
&\leq  e^{C \alpha^{-2} t}  \norm{\left( \mathcal{N} + \alpha^{-2} \right)^{5/2}  \Upsilon}_{\mathcal{F}}^2
\end{align}
which implies \eqref{eq: bound on the number operator} for $j=5$. Since $1\leq \alpha^{2j} (\mathcal{N}+\alpha^{-2})^j \leq \alpha^{10} (\mathcal{N}+\alpha^{-2})^5$ for $j\leq 5$, the general case follows.  
\end{proof}


\section{Proofs}

We shall restrict our attention to times $\vert t \vert \le \CL \alpha^2$,  where $T>0$ is chosen small enough such that the bounds in the previous section hold for some $0<\Lambda<\Lambda_0$. Note that $T$ can be chosen independent of $\alpha$. We shall also assume, without loss of generality, that $\alpha \geq \alpha_0$ for some $\alpha_0>0$; since the left side of \eqref{eq:thm1} is bounded by $2$, Theorem~\ref{thm:main} makes no claim for small $\alpha$. 
 
\subsection{Proof of Theorem \ref{thm:main}}
\label{subsection: proof main theorem}

We start by using Theorem \ref{thm:adiabatic} to estimate the error in replacing $\psi_t$ by $\psi_{\varphi_t}$ as 
\begin{align}
\label{eq: proof main theorem first step}
&  \| e^{-iHt}  \left( \psi_{\varphi_0} \otimes W( \alpha^2 \varphi_0)  \Upsilon \right)  - e^{-i \int_0^t ds \, \omega (s)} \psi_t  \otimes W( \alpha^2 \varphi_t ) \upt \| \notag \\[1.5mm]
 & \leq  \| e^{i \int_0^t ds \, \omega (s)} e^{-iHt}  \left( \psi_{\varphi_0} \otimes W( \alpha^2 \varphi_0) \Upsilon \right)  - e^{-i \int_0^t ds \, e( \varphi_s ) } \psi_{\varphi_t} \otimes W( \alpha^2 \varphi_t ) \upt \| + C \alpha^{-2} ,
\end{align} 
where $ e( \varphi)$ denotes the ground state energy of $h_{\varphi}$. Hence, our goal is to estimate the norm difference 
\begin{align}
\| e^{i \int_0^t ds \,  \omega (s)}&  e^{-iHt}  \left( \psi_{\varphi_0} \otimes  W( \alpha^2 \varphi_0 ) \Upsilon \right)  -  e^{-i \int_0^t ds \, e( \varphi_s ) } \psi_{\varphi_t}\otimes W( \alpha^2 \varphi_t ) \upt \| \notag\\[1mm]
& \quad =\| \xi_{t} - \psi_{\varphi_t} \otimes \upt \|,
\end{align}
where the fluctuation vector 
\begin{align}
\label{def:xi}
\xi_{t} =  e^{i \int_0^t ds \, \left( \omega (s) + e( \varphi_s) \right) } W^*( \alpha^2 \varphi_t) e^{-iHt} W( \alpha^2 \varphi_0) \left( \psi_{\varphi_0} \otimes \Upsilon \right)
\end{align}
satisfies
\begin{equation}\label{partL}
i \partial_t \xi_{t} = \mathcal{L}_t \xi_{t}
\end{equation}
with
\begin{equation}
\mathcal{L}_t  =  W^*\!\left( \alpha^2 \varphi_t \right) H W\!\left( \alpha^2 \varphi_t \right) + \left( i \partial_t W^*\!\left( \alpha^2 \varphi_t \right) \right) W\!\left( \alpha^2 \varphi_t \right) - \omega (t) - e( \varphi_t) .
\end{equation}
Since
\begin{align}
\label{eq:prop_Weyl}
W^*\!\left( f \right) a_k W\!\left( f \right) = a_k + \alpha^{-2} f( k), \quad W^*\!\left( f \right) a^*_k W\!\left( f \right) = a^*_k + \alpha^{-2}\overline{ f( k)}
\end{align}
and
\begin{align}
\label{eq:deriv_Weyl}
\left(i \partial_t W^*\!\left( \alpha^2 \varphi_t \right) \right) W\!\left( \alpha^2 \varphi_t \right) =  \alpha^2 \Im \langle \varphi_t, \partial_t \varphi_t \rangle - \phi \left( i \alpha^2 \partial_t \varphi_t \right)
\end{align}
(see, e.g., \cite[Lemma A.3]{FG}) the Landau--Pekar equations \eqref{eq:LP} imply that
\begin{equation}\label{def:L}
\mathcal{L}_t = h_{\varphi_t} - e( \varphi_t) + \mathcal{N} + \phi\left( \delta_t G_x\right) ,
\end{equation}
where we denote $\delta_t G_x = G_x - \sigma_{\psi_t}$. By the fundamental theorem of calculus, we obtain using Lemma \ref{lemma:minimizer}
\begin{align}
 \| \xi_{t} &- \psi_{\varphi_t} \otimes \upt \|^2  =  2 \Im \int_0^t ds \, \langle \xi_{s}, \, \left[ \mathcal{L}_s  - \mathcal{N} + \mathcal{A}_s  - i\alpha^{-2} R_s V_{i\varphi_s} \right] \psi_{\varphi_s} \otimes \ups \rangle .
\end{align}
Since $\left(h_{\varphi_s} - e( \varphi_s) \right) \psi_{\varphi_s}=0 $, this simplifies to  
\begin{align}
 \| \xi_{t} &- \psi_{\varphi_t} \otimes \upt \|^2  =  2 \Im \int_0^t ds \, \langle \xi_{s}, \, \left[ \phi \left( \delta_s G_x \right) + \mathcal{A}_s -  i\alpha^{-2} R_s V_{i\varphi_s} \right] \psi_{\varphi_s} \otimes \ups \rangle.
\end{align}
We insert the identity $1= p_s + q_s$, where $p_s = \vert \psi_{\varphi_s} \rangle \langle \psi_{\varphi_s} \vert$, to obtain
\begin{subequations}
\begin{align}
 \| \xi_{t} - \psi_{\varphi_t} \otimes \upt \|^2 
& =  2 \Im \int_0^t ds \, \langle \xi_{s}, \,p_s \phi \left( \delta_s G_x \right) \psi_{\varphi_s} \otimes \ups \rangle \label{eq:p}  \\
& \quad + 2 \Im \int_0^t ds \, \langle \xi_{s}, \, q_s  \left[  \phi \left(  G_x \right) - i \alpha^{-2} R_s V_{i\varphi_s} \right] \psi_{\varphi_s} \otimes \ups \rangle \label{eq:gap0}\\
& \quad + 2 \Im \int_0^t ds \, \langle\xi_{s}, \mathcal{A}_s \psi_{\varphi_s} \otimes \ups \rangle \label{eq:gapA}.
\end{align}
\end{subequations}
%
where we used that $q_s \, \phi \left( \sigma_{\psi_s} \right) \psi_{\varphi_s} =0$ implies  
\begin{equation}\label{eq: deltaphi between q and p}
q_s \, \phi \left( \delta_s G_x \right)\,  \psi_{\varphi_s} = q_s \,  \phi \left( G_x \right) \,  \psi_{\varphi_s} .
\end{equation}
 For the first term \eqref{eq:p}, we observe that
\begin{align}
p_s \phi \left( \delta_s G_x \right) \psi_{\varphi_s}&= \phi \left( \sigma_{\psi_s}- \sigma_{\psi_{\varphi_s}} \right) \psi_{\varphi_s} .
\end{align}
Since $\phi \left( \sigma_{\psi_s}- \sigma_{\psi_{\varphi_s}} \right)$ is a symmetric operator, we find 
\begin{align}
 \eqref{eq:p} & = 2 \int_0^t ds \, \Im \langle \xi_{s} - \psi_{\varphi_s} \otimes \ups, \, \phi \left(\sigma_{\psi_s} - \sigma_{\psi_{\varphi_s}} \right) \psi_{\varphi_s} \otimes \ups \rangle .
\end{align}
By using \eqref{eq: bounds for the creation and annihilation operators} this implies that 
\begin{equation}
\vert \eqref{eq:p} \vert \leq 4  \int_0^t ds \, \| \sigma_{\psi_s} -\sigma_{\psi_{\varphi_s}}\|_2 \|\left( \mathcal{N} + \alpha^{-2} \right)^{1/2} \ups \|  
\|   \xi_{s} - \psi_{\varphi_s} \otimes \ups \| .
\end{equation}
Since $\sigma_\psi$ depends only on $|\psi|^2$ and is independent of the phase of $\psi$, we can use 
Lemma \ref{lemma:Potental}  together with Lemmas \ref{lemma:LP} and \ref{lemma:minimizer}   and Theorem \ref{thm:adiabatic} to further bound 
\begin{align}
\label{eq:estimate_simga_diff}
\vert \eqref{eq:p} \vert & \leq  C \int_0^t ds \,\left( \| \psi_s \|_{H^1( \mathbb{R}^3)} + \| \psi_{\varphi_s} \|_{H^1( \mathbb{R}^3 )} \right) \| \psi_s - e^{-i \int_0^t ds \, e( \varphi_s )} \psi_{\varphi_s} \|_2 \notag\\
& \hspace{3cm}  \times \|\left( \mathcal{N} + \alpha^{-2} \right)^{1/2} \ups \|
\|   \xi_{s} - \psi_{\varphi_s} \otimes \ups \|
\notag  \\
& \leq  C \alpha^{-2} \int_0^t ds \, \| \left( \mathcal{N} + \alpha^{-2} \right)^{1/2} \ups \|
\|   \xi_{s} - \psi_{\varphi_s} \otimes \ups \|.
\end{align}
Applying in addition  Lemma \ref{lemma:bogo2} leads to the conclusion 
\begin{equation}
\vert \eqref{eq:p} \vert \leq  C \alpha^{-3} 
\int_0^t ds \, \|   \xi_{s} - \psi_{\varphi_s} \otimes \ups \|.
\end{equation}

In order to bound \eqref{eq:gap0}  we observe that \eqref{partL} and \eqref{def:L} imply that 
\begin{align}
\label{eq:gap_id}
q_s \xi_s = R_s \left[ h_{\varphi_s} - e( \varphi_s ) \right]  \xi_s = R_s \left[ i \partial_s - \mathcal{N} - \phi \left( \delta_s G_x \right) \right] \xi_s .
\end{align}
Hence we have
\begin{align}
\eqref{eq:gap0}   & =  2 \Im \int_0^t ds \, \langle  i \partial_s \xi_s, \, R_s   \left[ \phi \left( G_x \right) - i \alpha^{-2} R_s V_{i \varphi_s} \right] \psi_{\varphi_s} \otimes \ups \rangle \notag\\
& \quad - 2 \Im \int_0^t ds \, \langle \xi_s, \,  \left[\mathcal{N} + \phi \left( \delta_s G_x \right) \right]R_s \left[ \phi\left( G_x \right) - i \alpha^{-2} R_s V_{i \varphi_s} \right]  \psi_{\varphi_s} \otimes \ups \rangle    .
\end{align}
Integrating by parts in the first  line and recalling Definition \ref{def:bogo} we conclude that 
\begin{subequations}
\begin{align}
& \hspace{-4mm} \eqref{eq:gap0} + \eqref{eq:gapA} \notag \\
& = - 2 \Im \int_0^t ds \, \langle \xi_{s}, \left( \phi \left( \delta_s G_x \right) R_s \, \phi \left( G_x \right) - \mathcal{A}_s \right) \psi_{\varphi_s} \otimes \ups \rangle \label{eq:gap1,0}\\
&\quad  + 2 \alpha^{-2}  \Re \int_0^t ds \, \langle \xi_{s}, \left(  \phi  \left( \delta_s G_x \right)  R_s^2 \, V_{i\varphi_s} 
+ \alpha^2  \dot{R}_s\phi  \left( G_x \right)  \right) \psi_{\varphi_s} \otimes \ups \rangle  \label{eq:gap1,2b} \\
& \quad - 2 \Im \int_0^t ds \, \langle \xi_{s},R_s  \left( \mathcal{N}  \, \phi  \left( G_x \right) - \phi \left( G_x \right)\left( \mathcal{N}- \mathcal{A}_s\right) \right) \psi_{\varphi_s} \otimes \ups \rangle \label{eq:gap1,1} \\
& \quad + 2 \alpha^{-2}  \Re \int_0^t ds \, \langle \xi_{s},   R_s^2 \, V_{i\varphi_s}  \psi_{\varphi_s} \otimes \mathcal{A}_s\ups \rangle  \label{eq:gap1,2a} \\
& \quad + 2 \Re \int_0^t ds \, \langle \xi_{s}, \,  R_s \phi ( G_x) ( \partial_s \psi_{\varphi_s} ) \otimes \ups \rangle \label{eq:gap1,3} \\
& \quad + 2 \alpha^{-2} \Im \int_0^t ds \, \langle \xi_{s}, \,\left(  \left( \partial_s R_s^2 \right) V_{i\varphi_s} \psi_{\varphi_s} + R_s^2 V_{i \dot{\varphi}_s} \psi_{\varphi_s} + R_s^2 V_{i\varphi_s} (\partial_s \psi_{\varphi_s} ) \right) \otimes \ups \rangle\label{eq:gap1,4} \\[2mm]
 & \quad  - 2 \Re \langle \xi_{t}, R_t  \left[ \phi\left( G_x \right) - i \alpha^{-2} R_t V_{i \varphi_t} \right]  \psi_{\varphi_t} \otimes \upt \rangle  . \label{eq:gap1,5}
\end{align}
\end{subequations}
Here, we used that $R_0 \xi_0 = R_0 \psi_{\varphi_0} \otimes W( \alpha^2 \varphi_0 ) \Upsilon = 0$, hence the boundary terms at $t=0$ vanish. 

In the following, we shall bound the various terms on the right hand side of the previous equation. We start with \eqref{eq:gap1,1}. 
Using the CCR and $R_s = q_s R_s$, we find 
\begin{align}
 \eqref{eq:gap1,1}   & =  2  \Im \int_0^t ds \, \langle q_s \xi_{s}, \,   R_s \, a \left( G_x \right) \left( \alpha^{-2} + \mathcal{A}_s\right) \psi_{\varphi_s} \otimes \ups \rangle\notag \\
& \quad - 2  \Im \int_0^t ds \, \langle q_s  \xi_{s}  ,  \, R_s \, a^* \left( G_x \right) \left(\alpha^{-2} - \mathcal{A}_s \right) \psi_{\varphi_s} \otimes \ups \rangle ,
\end{align}
leading with Corollary \ref{cor:zsf_G} and Lemmas \ref{lemma:resolvent} and \ref{lemma:bogo} to
\begin{align}
\vert \eqref{eq:gap1,1} \vert &\leq  C \int_0^t ds \, \| \left( \mathcal{N} + \alpha^{-2}\right)^{3/2} \ups \|  \|q_s  \xi_s \| .
\end{align}
Since 
\begin{align}
\label{eq:qsxis}
\| q_s \xi_s \| \leq \| \xi_s - \psi_{\varphi_s} \otimes \ups \|
\end{align}
we obtain with Lemma \ref{lemma:bogo2} 
\begin{align}
\vert \eqref{eq:gap1,1} \vert
&\leq C \alpha^{-3} 
\int_0^t ds \, \| \xi_s - \psi_{\varphi_s} \otimes \ups \|.
\end{align}

With the same arguments and \eqref{cor:zsf_prelim}, we also conclude that 
\begin{align}
\vert \eqref{eq:gap1,2a}  \vert
&  \leq C \alpha^{-2} \int_0^t ds  \, \| \left( \mathcal{N} + \alpha^{-2} \right) \ups \| 
\| q_s \xi_{s}  \|
 \notag \\
& \leq C \alpha^{-4} \int_0^t ds \, \| \xi_s - \psi_{\varphi_s} \otimes \ups \| .
\end{align}
Similarly, by combining again \eqref{eq:qsxis}, Corollary \ref{cor:zsf_G} and Lemma \ref{lemma:bogo2} with Lemmas~\ref{lemma:minimizer} and~\ref{lemma:resolvent} and \eqref{cor:zsf_prelim}, we have
\begin{align}
\vert \eqref{eq:gap1,3} \vert & \leq C \alpha^{-2} \int_0^t ds \,   \| \psi_{\varphi_s} \|_{H^1 ( \mathbb{R}^3)} \| \left( \mathcal{N} +\alpha^{-2} \right)^{1/2} \ups \|
\| q_s \xi_s \|
\notag \\
& \leq C \alpha^{-3} \int_0^t ds \, \| \xi_s - \psi_{\varphi_s} \otimes \ups \|.
\end{align}
Since $\partial_s R_s^2 = \dot{R}_s R_s + R_s \dot{R}_s$ and $V_{i\dot{\varphi}_s} = \alpha^{-2} V_{ \varphi_s + \sigma_{\psi_s}}$ by \eqref{eq:LP}, we find with Lemmas \ref{lemma:Potental}, \ref{lemma:minimizer} and \ref{lemma:resolvent} that
\begin{align}
\vert \eqref{eq:gap1,4} \vert \leq C \alpha^{-4} |t| \le C   \alpha^{-2},
\end{align}
and similarly with \eqref{eq:qsxis}
\begin{align}
\vert \eqref{eq:gap1,5} \vert 
&\leq  2 \norm{q_t \xi_t } 
\Big(
\| R_t \phi ( G_x) \psi_{\varphi_t} \otimes \upt \|
+ \alpha^{-2} 
\|
R_t^2 V_{i\varphi_t} \psi_{\varphi_t} \otimes \upt
\| \Big)
\notag \\
&\leq \frac{1}{2} \norm{\xi_t - \psi_{\varphi_t} \otimes \upt}^2   + C \alpha^{-2} .
\end{align}

In order to estimate the first term \eqref{eq:gap1,0}, we insert again the decomposition  $1 = p_s + q_s$ and observe that
\begin{align}
p_s \phi \left( \delta_s G_x \right) R_s \, \phi  \left( G_x \right) \psi_{\varphi_s} \otimes \ups
&= \langle  \psi_{\varphi_s},  \phi \left(  G_{\,\cdot\,} \right) R_s \, \phi  \left( G_{\,\cdot\,} \right)  \psi_{\varphi_s} \rangle_{L^2(\mathbb{R}^3)} \,
\psi_{\varphi_s} \otimes \ups
\nonumber \\
&= \mathcal{A}_s \psi_{\varphi_s} \otimes \ups .
\end{align}
The Bogoliubov dynamics was in fact introduced in order to cancel this term. 
Hence
\begin{equation}
\eqref{eq:gap1,0}   = 2 \Im \int_0^t ds \, \langle \xi_{s}, \, q_s \phi \left( \delta_s G_x \right) R_s \phi  \left( G_x \right) \, \psi_{\varphi_s} \otimes \ups  \rangle .
\end{equation}
Recall that $\delta_s G_x = G_x - \sigma_{\psi_s}$. In the following, it will be convenient to replace $\psi_s$ by $\psi_{\varphi_s}$ in this expression, since the time derivative of the latter involves explicitly a factor $\alpha^{-2}$, see \eqref{partpsi}, leading to the bound \eqref{ps}. Hence  we shall write $\delta_s G_x = \widetilde{\delta}_s G_x + (\sigma_{\psi_{\varphi_s}} - \sigma_{\psi_s})$ with  $\widetilde{\delta}_s G_x = G_x - \sigma_{\psi_{\varphi_s}}$. For the second term, we use Lemma \ref{lemma:Potental}, Theorem \ref{thm:adiabatic} and the CCR  to bound  
\begin{align} \nonumber 
& \left| 2 \Im \int_0^t ds \, \langle \xi_{s}, \, q_s \phi \left( \sigma_{\psi_{\varphi_s}} - \sigma_{\psi_s} \right) R_s \phi  \left( G_x \right) \, \psi_{\varphi_s} \otimes \ups  \rangle \right|   \\ 
&\leq C  \int_0^t ds \, \| \sigma_{\psi_s} - \sigma_{\psi_{\varphi_s}} \|_2 \| \left( \mathcal{N} +\alpha^{-2} \right)^{1/2}  R_s \phi (G_x)\,  \psi_{\varphi_s} \otimes \ups \| \notag \\
&\leq  C \alpha^{-2} \int_0^t ds \,  \| R_s a(G_x)  \mathcal{N}^{1/2}  \psi_{\varphi_s} \otimes \, \ups \| \notag \\
& \quad + C \alpha^{-2} \int_0^t ds \,  \| R_s a^*(G_x) \left( \mathcal{N} + 2\alpha^{-2} \right)^{1/2} \,  \psi_{\varphi_s} \otimes \ups \|.
\end{align}
Corollary \ref{cor:zsf_G} and Lemma \ref{lemma:bogo2} thus imply that this term is bounded by $C\alpha^{-4} |t| \leq C \alpha^{-2}$. 

For the first term, we use once more \eqref{eq:gap_id} to obtain via integration by parts
\begin{subequations}
\begin{align} \nonumber
& 2 \Im \int_0^t ds \, \langle \xi_{s}, \, q_s \phi  \big( \widetilde{\delta}_s G_x \big) R_s \phi  \left( G_x \right) \, \psi_{\varphi_s} \otimes \ups \rangle 
\\ & =   -2 \Im \int_0^t ds \, \langle \xi_{s}, \, \mathcal{N} R_s  \phi  \left( \widetilde{\delta}_s G_x \right) R_s \phi \left( G_x \right) \, \psi_{\varphi_s} \otimes \ups \rangle \label{eq:gap2,1}
 \\
& \quad - 2 \Im \int_0^t ds \, \langle \xi_{s}, \,  \phi (\delta_s G_x)  R_s  \phi  \left( \widetilde{\delta}_s G_x \right) R_s \phi \left( G_x \right) \, \psi_{\varphi_s} \otimes \ups \rangle \label{eq:gap2,5}
\\
&\quad + 2 \Im \int_0^t ds \, \langle \xi_{s}, \, R_s \phi  \left( \widetilde{\delta}_s G_x \right) R_s \phi  \left( G_x \right) \left( \mathcal{N} -\mathcal{A}_s \right)\, \psi_{\varphi_s} \otimes \ups \rangle\label{eq:gap2,2} \\
& \quad + 2 \Re \int_0^t ds \, \langle\xi_{s}, \,  \left[ \partial_s \left( {R}_s \phi \left( \widetilde{\delta}_s G_x \right) R_s \phi  (G_x ) \psi_{\varphi_s}\right)\right] \otimes \ups  \rangle \label{eq:gap2,3}
\\
& \quad - 2 \Re \langle \xi_{t}, \, R_t \phi  \left( \widetilde{\delta}_t G_x \right) R_t \phi  \left( G_x \right) \psi_{\varphi_t} \otimes \upt \rangle \label{eq:gap2,4} .
\end{align}
\end{subequations}
With the aid of the CCR, the sum of the first and the third term can be rewritten as
\begin{align}
\eqref{eq:gap2,1} +  \eqref{eq:gap2,2} & =  4 \alpha^{-2} \Im \int_0^t ds \, \langle \xi_{s}, \,   R_s  a \left( \widetilde{\delta}_s G_x \right) R_s a\left( G_x \right)  \, \psi_{\varphi_s} \otimes \ups \rangle \notag
\\ & \quad  - 4 \alpha^{-2} \Im \int_0^t ds \, \langle \xi_{s}, \,   R_s  a^* \left( \widetilde{\delta}_s G_x \right) R_s a^*\left( G_x \right)  \, \psi_{\varphi_s} \otimes \ups \rangle \notag \\
& \quad - 2 \Im \int_0^t ds \, \langle \xi_{s}, \,  R_s  \phi  \left( \widetilde{\delta}_s G_x \right) R_s \phi \left( G_x \right) \mathcal{A}_s\, \psi_{\varphi_s} \otimes \ups \rangle .
\end{align}
Corollary \ref{cor:zsf_G}, Lemmas \ref{lemma:bogo} and \ref{lemma:bogo2} and a further application of the CCR thus imply that these terms are all bounded, in absolute value, by $C\alpha^{-4} |t| \leq C \alpha^{-2}$.  
For the forth term, we can use Lemma \ref{lemma:resolvent} to evaluate the derivative, which leads with the same arguments to the bound
\begin{align*}
\vert \eqref{eq:gap2,3} \vert \leq \alpha^{-4} |t| \le C \alpha^{-2}.
\end{align*}
Similarly, we also obtain
\begin{align}
\vert \eqref{eq:gap2,4} \vert \leq  C \alpha^{-2} .
\end{align}
We are thus left with giving a bound on \eqref{eq:gap2,5}. Note that since there is no resolvent to the left of $\phi(\delta_s G_x)$ in this term, we cannot proceed in the same way as with the other terms. Instead, we again insert the decomposition $1 = p_s + q_s$, use \eqref{eq: deltaphi between q and p} and the fact  that $\phi ( G_x)  R_s  \phi  ( \widetilde{\delta}_s G_x ) R_s \phi \left( G_x \right)$ is a symmetric operator in order to rewrite this term as 
\begin{subequations}
\begin{align}
\eqref{eq:gap2,5}
& =  
- 2 \Im \int_0^t ds \, \langle \xi_{s}, \, q_s  \phi (\delta_s G_x)  R_s  \phi  \left( \widetilde{\delta}_s G_x \right) R_s \phi \left( G_x \right) \, \psi_{\varphi_s} \otimes \ups \rangle 
\label{eq:gap2,5a}
\\
& \quad -
2 \Im \int_0^t ds \, \langle \xi_{s} - \psi_{\varphi_s} \otimes \ups, \, p_s  \phi ( G_x)  R_s  \phi  \left( \widetilde{\delta}_s G_x \right) R_s \phi \left( G_x \right) \, \psi_{\varphi_s} \otimes \ups \rangle .
\label{eq:gap2,5b}
\end{align}
\end{subequations}
Using Corollary \ref{cor:zsf_G} and Lemmas \ref{lemma:Potental}, \ref{lemma:minimizer}, \ref{lemma:resolvent} and \ref{lemma:bogo2}, we see that the second line is bounded by
\begin{align}
\abs{\eqref{eq:gap2,5b}}
&\leq C \int_0^t ds \, \| \psi_{\varphi_s} \|_{H^1( \mathbb{R}^3)}^2 \, \| \left( \mathcal{N} + \alpha^{-2}\right)^{3/2} \ups \|  
\|  \xi_s - \psi_{\varphi_s} \otimes \ups  \|
\notag \\
&\leq C \alpha^{-3} 
\int_0^t ds \, \|  \xi_s - \psi_{\varphi_s} \otimes \ups  \| .
\end{align}
For the first line, we introduce the shorthand notation
$\mathcal{C}_s = \phi ( \widetilde{\delta}_s G_x)  R_s  \phi  ( \widetilde{\delta}_s G_x ) R_s \phi \left( G_x \right)$. We again use \eqref{eq:gap_id} and integration by parts  to obtain
\begin{align}
\eqref{eq:gap2,5a}
&= - 2 \Im \int_0^t ds \, \langle \xi_{s}, \, q_s  \phi (\sigma_{\psi_{\varphi_s}} - \sigma_{\psi_s})  R_s  \phi  \left( \widetilde{\delta}_s G_x \right) R_s \phi \left( G_x \right) \, \psi_{\varphi_s} \otimes \ups \rangle 
\notag \\
&\quad
+ 2 \Im \int_0^t ds \,
\langle \xi_s ,  \phi \left( \delta_s G_x \right)  R_s \mathcal{C}_s \psi_{\varphi_s} \otimes \ups  \rangle
\notag \\
&\quad
+ 2 \Im \int_0^t ds \,
\langle \xi_s , \left[  \mathcal{N}  R_s \mathcal{C}_s - R_s \mathcal{C}_s \left( \mathcal{N} - \mathcal{A}_s \right)\right] \psi_{\varphi_s} \otimes \ups  \rangle
\notag \\
&\quad
- 2 \Re \int_0^t ds \,
\langle \xi_s , \left[  \dot{R}_s \mathcal{C}_s \psi_{\varphi_s}
+ R_s \dot{\mathcal{C}}_s   \psi_{\varphi_s}
+ R_s \mathcal{C}_s ( \partial_s \psi_{\varphi_s} ) \right] \otimes \ups  \rangle
\notag \\[2.5mm]
&\quad +
2 \Re
\langle\xi_t, R_t \mathcal{C}_t \psi_{\varphi_t} \otimes \upt \rangle.
\end{align}
All terms except the one in the second line can be dealt with in the same way as before, leading to a bound of the order $\alpha^{-3}$. For the second line, we shall use energy conversation to argue that $\langle \xi_s , (1-\Delta) \xi_s\rangle$ is uniformly bounded. In fact, this follows because
\begin{align}
\label{eq:Econs}
-\Delta \leq  C (H + C)  \leq C (- \Delta + \mathcal{N} + 1)
\end{align}
uniformly in $\alpha \geq \alpha_0$  (see \cite{liebthomas,liebyamazaki} or  \cite[Lemma~A.3]{LRSS}) and since $\langle \Psi_0 , (- \Delta + \mathcal{N} + 1) \Psi_0 \rangle \leq C$ by the assumptions of Theorem \ref{thm:main}.
We can thus bound
\begin{equation}
\left| \langle \xi_s ,  \phi \left( \delta_s G_x \right)  R_s \mathcal{C}_s \psi_{\varphi_s} \otimes \ups  \rangle \right| \leq C \left\| (1-\Delta)^{-1/2}  \phi \left( \delta_s G_x \right)  R_s \mathcal{C}_s \psi_{\varphi_s} \otimes \ups\right\| \leq C \alpha^{-4} ,
\end{equation}
where we used again Corollary~\ref{cor:zsf_G} in the last step. In combination, we thus have
\begin{equation}
\abs{\eqref{eq:gap2,5a}} \leq C \alpha^{-3}  + C \alpha^{-4} \abs{t}  \leq C \alpha^{-2} .
\end{equation}
This completes the derivation of the bound for \eqref{eq:gap1,0}, which reads
\begin{align}
\vert \eqref{eq:gap1,0} \vert  \le C \alpha^{-3} 
\int_0^t ds \, \|  \xi_s - \psi_{\varphi_s} \otimes \ups  \|  + C \alpha^{-2}.
\end{align}

It remains to bound \eqref{eq:gap1,2b}. Using  the expression \eqref{eq: derivative rho} for $\dot{R}_s$, it is given by 
\begin{align}
\eqref{eq:gap1,2b} &= - 2 \alpha^{-2} \int_0^t ds \,
\Re \langle \xi_s , 
\Big( R_s^2 V_{i \varphi_s} p_s  
- R_s \left( V_{i \varphi_s} - \langle \psi_{\varphi_s}, V_{i \varphi_s} \psi_{\varphi_s} \rangle \right) R_s \Big) \phi(G_x) \psi_{\varphi_s} \otimes \ups \rangle
\notag \\
& \quad + 2 \alpha^{-2} \int_0^t ds \,
\Re \langle \xi_s , \Big( \phi (\delta_s G_x ) R_s^2 V_{i \varphi_s} - p_s V_{i \varphi_s} R_s^2 \phi(G_x) \Big) \psi_{\varphi_s} \otimes \ups \rangle  .
\end{align}
The first line can be bounded in the same way as before, by $C \alpha^{-3} \int_0^t ds \| \xi_s - \psi_{\varphi_s}\otimes\ups \|$. 
In the second line, we insert  the decomposition $1 = p_s + q_s$ and arrive at the terms
\begin{subequations}
\begin{align}\label{eq:gap1,2b1}
&  2 \alpha^{-2} \int_0^t ds \,
\Re \langle \xi_s - \psi_{\varphi_s} \otimes \ups, p_s \left( \phi (\delta_s G_x ) R_s^2 V_{i \varphi_s} - V_{i \varphi_s} R_s^2 \phi(G_x) \right) \psi_{\varphi_s} \otimes \ups \rangle
\\
\label{eq:gap1,2b3}
&  + 2 \alpha^{-2} \int_0^t ds \,
\Re \langle \xi_s  , q_s \phi (\sigma_{\psi_{\varphi_s}} - \sigma_{\psi_s} ) R_s^2 V_{i \varphi_s}  \psi_{\varphi_s} \otimes \ups \rangle
\\
\label{eq:gap1,2b4}
&  + 2 \alpha^{-2} \int_0^t ds \,
\Re \langle \xi_s , q_s \phi (\widetilde{\delta}_s G_x ) R_s^2 V_{i \varphi_s}  \psi_{\varphi_s} \otimes \ups \rangle .
\end{align}
\end{subequations}
In the first term, we used  \eqref{eq: deltaphi between q and p}  and the fact that the expectation value of $ \phi ( G_x ) R_s^2 V_{i \varphi_s} - V_{i \varphi_s} R_s^2 \phi(G_x)$ in the state $\psi_{\varphi_s} \otimes \ups $ is purely imaginary in order to replace $\xi_s$ by $\xi_s- \psi_{\varphi_s} \otimes \ups$. In the last term, we use again the notation  $\widetilde{\delta}_s G_x = G_x - \sigma_{\psi_{\varphi_s}} = \delta_s G_x + \sigma_{\psi_s}- \sigma_{\psi_{\varphi_s}}$.
Analogous estimates as before 
show that the first two lines can be bounded by
\begin{align}
\abs{\eqref{eq:gap1,2b1}}  
+ \abs{\eqref{eq:gap1,2b3}}
&\leq C \alpha^{-3} \int_0^t ds \, 
\|  \xi_s - \psi_{\varphi_s} \otimes \ups \|.
\end{align}
For the last term, we use once more \eqref{eq:gap_id}  and integration by parts to obtain
\begin{align}
\eqref{eq:gap1,2b4} & =   -
2 \alpha^{-2}
\int_0^t ds \, 
\Re \langle \xi_s, R_s \left( \mathcal{N} \phi(\widetilde{\delta}_s G_x) - \phi (\widetilde{\delta}_s G_x) (\mathcal{N} - \mathcal{A}_s) \right) R_s^2 V_{i \varphi_s} \psi_{\varphi_s} \otimes \ups \rangle 
\notag \\
&\quad -
2 \alpha^{-2} \int_0^t ds \, \Re
\langle \xi_s,  \phi(\delta_s G_x )  R_s \phi(\widetilde{\delta}_s G_x ) R_s^2 V_{i \varphi_s} \psi_{\varphi_s} \otimes \ups \rangle
\notag \\
&\quad
- 2 \alpha^{-2} \int_0^t ds \,
\Im \langle \xi_s, \left[ \partial_s \left( R_s \phi(\widetilde{\delta}_s G_x ) R_s^2 V_{i \varphi_s} \psi_{\varphi_s} \right) \right] \otimes \ups \rangle  
\notag \\[2.5mm]
& \quad  - 2 \alpha^{-2} \Im 
\langle \xi_t, R_t \phi (\widetilde{\delta}_t G_x) R_t^2 V_{i \varphi_t} \psi_{\varphi_t} \otimes \upt \rangle .
\end{align}
In the same way as before, using \eqref{eq:Econs} for the second term, 
we obtain
\begin{align}
\abs{\eqref{eq:gap1,2b4}}
&\leq C \alpha^{-3} +  C \alpha^{-5} \abs{t} \le C \alpha^{-3},
\end{align}
and thus 
\begin{align}
\vert \eqref{eq:gap1,2b} \vert \le C   \alpha^{-3} \int_0^t ds \, 
\|  \xi_s - \psi_{\varphi_s} \otimes \ups \| + C \alpha^{-2}.
\end{align}

Collecting all the bounds, we have thus proved that  
\begin{align}
\| \xi_{t} - \psi_{\varphi_t} \otimes \upt \|^2
&\leq C \alpha^{-2} 
+ C \alpha^{-3} \int_0^t ds \, \| \xi_{s} - \psi_{\varphi_s} \otimes \ups \|.
\end{align}
Gr\"onwall's inequality then leads to
\begin{align}
\| \xi_{t} - \psi_{\varphi_t} \otimes \upt \|^2
\leq C  \alpha^{-2} \left( 1 + \alpha^{-4} \abs{t}^2 \right) \le C\alpha^{-2}.
\end{align}
In combination with \eqref{eq: proof main theorem first step} this completes the proof of Theorem \ref{thm:main}.

\subsection{Proof of Theorem \ref{corollary:reduced densities} }
\label{subsection: proof corollary reduced densities}

The bound \eqref{eq: bdelectron density} for the electron reduced density matrix follows from
\begin{align}
&\left \| \gamma^\text{el}_t - \ket{\psi_t} \bra{\psi_t} \right\|_{\rm tr}
\nonumber \\
&\quad =  \textnormal{Tr}_{L^2({\mathbb{R}^3)}}
\left|  \textnormal{Tr}_{\mathcal{F}} \left(
\ket{e^{-i H t} \Psi_0} \bra{e^{- i H t} \Psi_0}
- \ket{\psi_t \otimes W \left (\alpha^2 \varphi_t \right) \upt}
\bra{\psi_t \otimes W \left (\alpha^2 \varphi_t \right) \upt} \right)
\right|
\notag \\
&\quad \leq 
\textnormal{Tr}_{\mathcal{H}}
\left| 
\ket{e^{-i H t} \Psi_0} \bra{e^{- i H t} \Psi_0}
- \ket{\psi_t \otimes W \left (\alpha^2 \varphi_t \right) \upt}
\bra{\psi_t \otimes W \left (\alpha^2 \varphi_t \right) \upt} \,
\right|
\notag \\
&\quad \leq  2
\norm{e^{-i H t} \Psi_0 - e^{- i \int_0^t ds \, \omega(s)} \psi_t \otimes W(\alpha^2 \varphi_t) \upt}
\end{align}
in combination with Theorem \ref{thm:main}.
In order to
prove \eqref{eq: bdphonon density}, we start by noting that 
\begin{align}
\left( \gamma_t^{\rm ph} - \ket{\varphi_t} \bra{\varphi_t} \right)(k,l)
&= \scp{W^*\!\left( \alpha^2 \varphi_t \right) e^{- i H t} \Psi_0}{a_l^* a_k W^*\!\left( \alpha^2 \varphi_t \right) e^{- i H t} \Psi_0}
\nonumber \\
&\quad 
+ \varphi_t(k) \scp{W^*\!\left( \alpha^2 \varphi_t \right) e^{- i H t} \Psi_0}{a_l^* W^*\!\left( \alpha^2 \varphi_t \right) e^{- i H t} \Psi_0}
\nonumber \\[2mm]
&\quad
+ \overline{\varphi_t(l)}
\scp{W^*\!\left( \alpha^2 \varphi_t \right) e^{- i H t} \Psi_0}{a_k W^*\!\left( \alpha^2 \varphi_t \right) e^{- i H t} \Psi_0}
\end{align}
using  \eqref{eq:prop_Weyl} and unitarity of the Weyl operators. The first term defines a positive operator, hence its trace norm equals its trace, given by $\| \mathcal{N}^{1/2} W^*( \alpha^2 \varphi_t) e^{- i H t} \Psi_0\|^2$. The other two terms define operators of rank one. Using the triangle inequality for the trace norm, as well as the Cauchy--Schwarz inequality to bound the rank one terms, we conclude that
\begin{equation}
 \left\|  \gamma_t^{\rm ph} - \ket{\varphi_t} \bra{\varphi_t} \right\|_{\rm tr} \leq  
 \left\| \mathcal{N}^{1/2} W^*\!\left( \alpha^2 \varphi_t \right) e^{- i H t} \Psi_0  \right\|^2 + 2 \| \varphi_t \|_2
\left\| \mathcal{N}^{1/2} W^*\!\left( \alpha^2 \varphi_t \right) e^{- i H t} \Psi_0  \right\| .
\end{equation}
Let $\alpha_0>0$. For small $\alpha\leq\alpha_0$, we use the form bound
\begin{align}
\label{eq:formbound_H}
\mathcal{N} \leq C \left( H + C \alpha^{-2} \right) \leq C \left( - \Delta+ \mathcal{N} + \alpha^{-2} \right) 
\end{align}
for whose proof we refer to the commutator method of Lieb and Yamazaki \cite{liebyamazaki}, see also \cite{liebthomas} or \cite[Lemma 7]{FS}. In fact, \eqref{eq:formbound_H} implies the trivial bound
\begin{align}
\left\|  \gamma_t^{\rm ph} - \ket{\varphi_t} \bra{\varphi_t} \right\|_{\rm tr} \leq C \left( 1 + \alpha^{-2} \right) ,
\end{align}
for all $\alpha \leq \alpha_0$. 
The bound \eqref{eq: bdphonon density} is then a consequence of \eqref{eq:bdN} for $\alpha>\alpha_0$, whose proof occupies the rest of this section. Hence, in the rest of this section, we restrict to $\alpha > \alpha_0$. 


We split 
\begin{equation}\label{eq: bound for split N in proof}
\| \mathcal{N}^{1/2} W^*( \alpha^2 \varphi_t ) e^{-iHt} \Psi_0   \|^2 = \| \mathcal{N}^{1/2}_\leq  W^*( \alpha^2 \varphi_t ) e^{-iHt} \Psi_0   \|^2 +  \| \mathcal{N}^{1/2}_> W^*( \alpha^2 \varphi_t ) e^{-iHt} \Psi_0   \|^2 
\end{equation}
where, for $K >0$, we write $\mathcal{N} = \mathcal{N}_\leq + \mathcal{N}_>$, with
\begin{align}
 \quad \mathcal{N}_{\leq} = \int_{|k| \leq K} dk \, a_k^*a_k .
\end{align}
To bound the right side of \eqref{eq: bound for split N in proof}, we make use of the following lemma, which is proven at the end of this section.

\begin{lemma}\label{lemma:Nsplit}
Let $\alpha_0 > 0$. Under the same assumptions as in Theorem \ref{thm:main}, there exist $C,T >0$ such that
\begin{align}
\| \mathcal{N}_\leq W^* (\alpha^2 \varphi_t ) e^{-iHt} \Psi_0 \| \leq C   \left( 1 + K^{1/2} \right) \label{eq:boundNleq} 
\end{align}
for all $K>0$, $\alpha \geq \alpha_0$ and $\vert t \vert \le \CL \alpha^2$. Moreover, under the additional assumption that $\varphi_0 \in L^2_{1/4} (\mathbb{R}^3 ) : = L^2 ( \mathbb{R}^3, \, ( 1+|k|^2)^{1/4} dk )$, we have
\begin{align}
\| \mathcal{N}_>^{1/2} W^* (\alpha^2 \varphi_t ) e^{-iHt}  \Psi_0 \|^2 \leq C\left(  K^{-1/2}   +  \alpha^{-2} \right)  \label{eq:boundNgeq}
\end{align}
for all $K>0$, $\alpha \geq \alpha_0$ and $\vert t \vert \le \CL \alpha^2$.
\end{lemma} 

Thus, writing
\begin{align}\nonumber
& \| \mathcal{N}^{1/2}_\leq  W^*( \alpha^2 \varphi_t ) e^{-iHt} \Psi_0 \|^2  
\\
& = \left\langle  \left( W^*( \alpha^2 \varphi_t ) e^{-iHt} \Psi_0  - e^{-i\int_0^t ds \, \omega(s)} \psi_t \otimes \upt \right) , \mathcal{N}_\leq W^*( \alpha^2 \varphi_t ) e^{-iHt} \Psi_0  \right\rangle\notag \\[1mm]
&    \quad +  e^{i\int_0^t ds \, \omega(s)} \langle \psi_t \otimes \upt , \, \mathcal{N}_\leq \,  W^*( \alpha^2 \varphi_t ) e^{-iHt} \Psi_0  \rangle 
\label{eq:N12-1}
\end{align}
implies that
\begin{align}\nonumber
& \| \mathcal{N}^{1/2}_\leq W^*( \alpha^2 \varphi_t ) e^{-iHt} \Psi_0   \|^2 \\
&  \leq \| (W^*( \alpha^2 \varphi_t ) e^{-iHt} \Psi_0 - e^{-i\int_0^t ds \, \omega(s)} \psi_t \otimes \upt \| \,  \| \mathcal{N}_\leq W^* ( \alpha^{2} \varphi_t ) e^{-iHt} \Psi_0  \|+ \| \mathcal{N}_\leq \psi_t \otimes \upt \|  ,\label{eq:N12-2}
\end{align}
which leads with Theorem \ref{thm:main} and Lemma \ref{lemma:Nsplit} for the first term,  and Lemma \ref{lemma:bogo2} for the second term, to
\begin{align}
\| \mathcal{N}^{1/2}_\leq W^*( \alpha^2 \varphi_t ) e^{-iHt} \Psi_0  \|^2 \leq C \alpha^{-1} \left( 1 + K^{1/2}  \right) +  C \alpha^{-2}.
\end{align}
In combination with \eqref{eq:boundNgeq}, we thus have
\begin{align}
\| \mathcal{N}^{1/2} W^* (\alpha^2 \varphi_t ) e^{-iHt}   \Psi_0  \|^2 \leq C \left( \alpha^{-2} + \alpha^{-1} ( 1 + K^{1/2}  )+ K^{-1/2}  \right) .
\end{align}
The choice $K = \alpha $ leads to \eqref{eq:bdN} and hence 
completes the proof of Theorem \ref{corollary:reduced densities}.
\hfill$\square$

\bigskip

For the proof of Lemma \ref{lemma:Nsplit} we need the following statement.

\begin{lemma}
\label{lemma:N}
Let $\alpha_0 >0$. Under the same assumptions as in Theorem \ref{thm:main}, there exist $C,\CL >0$ such that
\begin{align}
\| \mathcal{N}^{1/2} W^* (\alpha^2 \varphi_t ) e^{-iHt}   \Psi_0  \| \leq C .
\end{align}
for all $|t| \leq \CL \alpha^2$ and $\alpha \geq \alpha_0$. 
\end{lemma}

\begin{proof}[Proof of Lemma \ref{lemma:N}] Recall the definition of the fluctuation vector $\xi_t$ in \eqref{def:xi} 
satisfying $i \partial_t \xi_{t}  =\mathcal{L}_t \xi_{t}$ with $\mathcal{L}_t$  given in \eqref{def:L}.
We have
\begin{align}
\| \mathcal{N}^{1/2} W^* (\alpha^2 \varphi_t ) e^{-iHt}    \Psi_0  \| = \| \mathcal{N}^{1/2} \xi_t \|,
\end{align}
and using the CCR we compute
\begin{align}
\| \mathcal{N}^{1/2} \xi_t \|^2  - \| \mathcal{N}^{1/2} \xi_0 \|^2 & = i \int_0^t ds \, \langle \xi_s, \, \left[ \mathcal{L}_s, \mathcal{N} \right]\xi_s \rangle \notag \\
&  = i \alpha^{-2} \int_0^t ds \, \langle\xi_s , \, \left[ a( \delta_s G_x) - a^*(\delta_s G_x) \right]\xi_s \rangle ,
\end{align}
where $\delta_s G_x = G_x - \sigma_{\psi_s}$. 
Thus,
\begin{align}
\| \mathcal{N}^{1/2} \xi_t \|_2^2  - \| \mathcal{N}^{1/2} \xi_0 \|^2 & \leq 2 \alpha^{-2} \int_0^t ds \, \| (- \Delta +1	)^{-1/2} a^*( G_x ) \xi_s \| \, \| (- \Delta +1)^{1/2} \xi_s \|  \notag \\
& \quad + 2 \alpha^{-2} \int_0^t ds \, \| \sigma_{\psi_s} \|_2 \| \mathcal{N}^{1/2} \xi_s \| \, \|\xi_s \| .
\end{align}
Using \eqref{eq:Econs},  Lemmas \ref{lemma:Potental} and \ref{lemma:G}  we find 
\begin{align}
\| \mathcal{N}^{1/2} \xi_t \|_2^2  - \| \mathcal{N}^{1/2} \xi_0 \|^2 
\leq  C \alpha^{-2} \int_0^t ds \, \| \left( \mathcal{N} + \alpha^{-2} \right)^{1/2} \xi_s \|  .
\end{align}
Since $\| \mathcal{N}^{1/2} \xi_0 \|  = \| \mathcal{N}^{1/2} \Upsilon  \|_{\mathcal F}  \le C \alpha^{-1}$ by assumption, we conclude with Gr\"onwall's inequality that
\begin{align}
\| (\mathcal{N}  +  \alpha^{-2} )^{1/2} \xi_t \| \leq C \left( \alpha^{-2}  |t| +  \alpha^{-1} \right)  \le C  .\label{eq:GW_sqrtN}
\end{align}
\end{proof}



\begin{proof}[Proof of Lemma \ref{lemma:Nsplit}] We use again the notation introduced in \eqref{def:xi} and \eqref{def:L}, and compute 
\begin{align}
\| \mathcal{N}_\leq \xi_t \|^2 - \| \mathcal{N}_\leq \xi_0 \|^2 & = i \int_0^t ds \,  \langle \xi_s, \left[ \mathcal{L}_s, \, \mathcal{N}_\leq^2 \right] \xi_s \rangle .
\end{align}
Since
\begin{align}
\left[ \mathcal{L}_s, \mathcal{N}_\leq^2 \right] = \mathcal{N}_\leq \left[ \mathcal{L}_s,\,  \mathcal{N}_\leq \right] + \left[ \mathcal{L}_s,\,  \mathcal{N}_\leq \right] \mathcal{N}_\leq 
\end{align}
and
\begin{equation}
\left[ \mathcal{L}_s, \, \mathcal{N}_\leq \right]   = \alpha^{-2} \int_{|k| \leq K} dk \left[ \left( |k|^{-1} e^{ik \cdot x} - \overline{ \sigma_{\psi_s}(k)} \right)  a_k - \left( |k|^{-1} e^{-ik \cdot x} - \sigma_{\psi_s}(k) \right) a_k^*   \right] ,
\end{equation}
we have 
\begin{equation}
\|  \mathcal{N}_\leq \xi_t \|^2  - \| \mathcal{N}_\leq \xi_0 \|^2 \leq C \alpha^{-2} \int_0^t ds \, \left(  K^{1/2} +  \| \sigma_{\psi_s} \|_2 \right) \, \| \left( \mathcal{N}_\leq + \alpha^{-2} \right)^{1/2} \xi_s \| \, \| \mathcal{N}_\leq  \xi_s \|.
\end{equation}
Lemma \ref{lemma:N} implies that $ \| ( \mathcal{N}_\leq + \alpha^{-2})^{1/2} \xi_s \| \leq C$ for all $|s| \leq \CL \alpha^2$. Hence we obtain with Lemma \ref{lemma:Potental} and Gr\"onwall's inequality
\begin{align}
\| \mathcal{N}_\leq \xi_t \| \leq \| \mathcal{N}_\leq \xi_0 \| + C \alpha^{-2} \left( 1 + K^{1/2} \right) |t| \leq \| \mathcal{N}_\leq \xi_0 \| + C  \left( 1 + K^{1/2} \right)   ,
\end{align}
leading with the assumption $\| \mathcal{N}_\leq \xi_0 \| = \| \mathcal{N}_\leq  \Upsilon \|_{\mathcal F}   \leq C \alpha^{-2}$ to 
\begin{align}
\| \mathcal{N}_\leq \xi_t \| \leq C   \left(  1 + K^{1/2} \right) . 
\end{align}

In order to prove \eqref{eq:boundNgeq}, we first derive a bound on the $L^2_{1/4} ( \mathbb{R}^3)$-norm of $\varphi_t$. 
For that purpose, let $\eta_t (k) = |k|^{1/4} \varphi_t (k)$. The Landau--Pekar equations \eqref{eq:LP} imply 
\begin{equation}
\alpha^2\partial_t \| \eta_t \|_2^2  =  2 \Im \int dk \, |k|^{1/4} \eta_{t}(k ) \, \overline{\sigma_{\psi_t} (k) }.
\end{equation}
With the aid of the Cauchy--Schwarz inequality we thus obtain
\begin{align}
\| \eta_t \|_2^2 - \| \eta_0 \|_2^2 \leq  \alpha^{-2} \int_0^t ds \, \left( \| \eta_s \|_2^2 +  \| |\, \cdot\, |^{1/4} \sigma_{\psi_s} \|_{2} ^2 \right) .
\end{align}
Since for arbitrary $\kappa >0$
\begin{equation}
\| |\, \cdot \,|^{1/4} \sigma_{\psi_s} \|_{2}^2   \leq   \kappa^{1/2}  \| \sigma_{\psi_s} \|_2^2 + \kappa^{-3/2}   \|  | \cdot |   \sigma_{\psi_s} \|_2^2 , 
\end{equation}
%
and by the Plancherel identity and Sobolev's inequality,
\begin{align}
\| \vert \cdot \vert \sigma_{\psi_s} \|_2 = C \| \vert \psi_s \vert^2 \|_2 = C \| \psi_s \|_4^4 \leq C \| \psi_s \|_{H^1( \mathbb{R}^3)}^4,
\end{align}
we conclude with Lemmas~\ref{lemma:LP} and~\ref{lemma:Potental}  and the assumption $\varphi_0 \in L_{1/4}^2 \left( \mathbb{R}^3 \right)$  that
\begin{align}
\label{eq:bound_phi_weight}
\| \varphi_t \|_{L^2_{1/4} \left( \mathbb{R}^3\right)} \leq C 
\end{align}
for all $\vert t \vert \leq \CL \alpha^2$.

Using \eqref{eq:prop_Weyl} we compute
\begin{align}
W\!\left( \alpha^2 \varphi_t \right) \mathcal{N}_> W^*\!\left( \alpha^2 \varphi_t \right)  = \mathcal{N}_> - \phi \left( \chi \left( \vert \cdot \vert \geq K \right) \varphi_t \right) + \| \chi \left( \vert \cdot \vert \geq K \right) \varphi_t \|_2^2 . \label{eq:WeylN}
\end{align}
Thus, writing $\widetilde{\xi}_t = W\!\left( \alpha^2 \varphi_t \right) \xi_t = e^{i \int_0^t ds (\omega(s) + e(\varphi_s)) } e^{-iHt} \, \Psi_0 $, we find 
\begin{align}
\| \mathcal{N}_>^{1/2} \xi_t \|^2 & = \langle \widetilde{\xi}_t , \, W\!\left( \alpha^2 \varphi_t \right) \mathcal{N}_> W^*\!\left( \alpha^2 \varphi_t \right)  \widetilde{\xi}_t \rangle \leq 2 \| \mathcal{N}_>^{1/2} \widetilde{\xi}_t\|^2 + 2 \| \chi \left( \vert \cdot \vert \geq K \right) \varphi_t \|_2^2 .
\end{align}
Since $\| \chi \left( \vert \cdot \vert \geq K \right) \varphi_t \|_2 \leq K^{-1/4} \| \varphi_t \|_{L^2_{1/4} \left( \mathbb{R}^3\right)}$ it follows from
%
\eqref{eq:bound_phi_weight} that
\begin{align}\label{eq: relation between ksi and tilde ksi}
\| \mathcal{N}_>^{1/2} \xi_t \|^2   \leq 2  \, \| \mathcal{N}_>^{1/2} \widetilde{\xi}_t\|^2 + C K^{-1/2} .
\end{align}
Note that $\widetilde{\xi}_0 = \Psi_0$. Hence, it remains to estimate
\begin{align}
\| \mathcal{N}_>^{1/2} & \widetilde{\xi}_t\|^2 - \| \mathcal{N}_>^{1/2}  \Psi_0 \|^2 = i \int_0^t ds \, \langle \widetilde{\xi}_s,  \left[ H, \, \mathcal{N}_> \right] \widetilde{\xi}_s\rangle \label{eq:boundNgeq1}
\end{align}
where
\begin{align}
\left[ H, \, \mathcal{N}_> \right] & = \alpha^{-2} \int_{|k| > K}   \frac{dk}{|k|}\left( e^{ik \cdot x} a_k - e^{-ik \cdot x} a_k^*  \right) .
\end{align}
It follows again from the commutator method of Lieb and Yamazaki \cite{liebyamazaki}, (resp. \cite{liebthomas} or \cite[Lemma 7]{FS}) that
\begin{equation}
i \left[ H, \, \mathcal{N}_> \right] \leq C \alpha^{-2} K^{-1/2} ( H + C ) \,.
\end{equation}

And we have with \eqref{eq:Econs} that $\langle \widetilde{\xi}_s, \, H \widetilde{\xi}_s \rangle = \langle \Psi_0 , \, H   \Psi_0 \rangle \leq C$ by assumption. 
%
We thus conclude that 
\begin{align}
\| \mathcal{N}_>^{1/2} & \widetilde \xi_t \|^2 - \| \mathcal{N}_>^{1/2} \Psi_0 \|^2\le C \vert t \vert \alpha^{-2} K^{-1/2} \le C   K^{-1/2} . 
\end{align}
Using again \eqref{eq:prop_Weyl}  we find similarly as in \eqref{eq:WeylN}
\begin{align}
\| \mathcal{N}_>^{1/2}  \Psi_0 \|^2 & = \| \mathcal{N}_>^{1/2}   W\!\left( \alpha^2 \varphi_0 \right) \Upsilon  \|^2_{\mathcal F}
\notag \\[1mm]
&  \leq 2 \| \mathcal{N}^{1/2}_> \Upsilon \|^2_{\mathcal F} + 2 K^{-1/2} \| \varphi_0 \|_{L^2_{1/4}( \mathbb{R}^3)}^2 \leq C\left( \alpha^{-2} + K^{-1/2} \right)
\end{align}
where we used the assumptions $\varphi_0 \in L^2_{1/4}(\mathbb R^3)$ and $\| \mathcal{N}^{1/2}_> \Upsilon \|^2_{\mathcal F}\le c \alpha^{-2}$. Hence we obtain 
\begin{align}
\| \mathcal{N}_>^{1/2}  \widetilde \xi_t \|^2 \leq C \left( K^{-1/2} + \alpha^{-2} \right)
\end{align}
and with \eqref{eq: relation between ksi and tilde ksi} therefore also the desired bound \eqref{eq:boundNgeq}. 
\end{proof}

\subsection{Proof of Remark \ref{rmk:B} }
\label{rmk:bogo}

\label{subsection: proof of remark bogo}

It suffices to show that for $\Upsilon_0 = \Omega$ and sufficiently small $\delta > 0$, there exists a constant $C_\delta > 0$ such that for $t= \delta \alpha^2$
\begin{align}
 \| \upt - \Omega \| \geq C_\delta
\end{align}
uniformly in $\alpha$.
To this end, we compute
\begin{align}
\upt - \Omega &= \int_0^t ds \, \frac{d}{ds} \ups = -i \int_0^t ds \, (\mathcal{N} - \mathcal{A}_s ) \ups \notag\\
 &=  i t \mathcal{A}_0 \Omega - i \int_0^t ds \int_0^s  d\tau \frac{d}{d\tau} (\mathcal{N} - \mathcal{A}_\tau ) \Upsilon_{\tau} \notag\\
&= i t \mathcal{A}_0 \Omega -  \int_0^t ds \int_0^s d\tau \, (\mathcal{N} - \mathcal{A}_\tau )^2 \Upsilon_{\tau} + i \int_0^t ds \int_0^s d\tau \, \dot{A}_\tau \Upsilon_{\tau} ,
\end{align}
and hence
\begin{align}
\langle \Omega, \upt \rangle - 1 = it \langle \Omega, \mathcal{A}_0 \Omega \rangle - \int_0^t ds \int_0^s d\tau \,  \langle \Omega, (\mathcal{N} - \mathcal{A}_\tau )^2 \Upsilon_{\tau} \rangle  + i \int_0^t ds \int_0^s d\tau \, \langle \Omega , \dot{A}_\tau \Upsilon_{\tau} \rangle  .
\end{align}
Lemma \ref{lemma:bogo} 
implies that 
\begin{align}
\left| \langle \Omega, (\mathcal{N} - \mathcal{A}_\tau )^2 \Upsilon_{\tau} \rangle  \right| \leq \| (\mathcal{N} - \mathcal{A}_\tau)^2 \Omega  \| \leq C \alpha^{-4}  \quad \text{and } \quad \left| \langle \Omega , \dot{A}_\tau \Upsilon_{\tau} \rangle \right| \leq C \alpha^{-4} ,
\end{align}
and with the notation introduced in \eqref{def:F}
\begin{align}
 \langle \Omega, \mathcal{A}_0 \Omega \rangle = \alpha^{-2} \int dk \, F_0 (k, k) = \alpha^{-2} \int dk \,  \frac{1}{|k|^2} \left\| R_0^{1/2} e^{ik \, \cdot \, } \psi_{\varphi_0} \right\|_2^2 =: c_0 \alpha^{-2} 
 \end{align}
for $c_0 > 0$. We thus conclude that 
\begin{align}
\| \upt - \Omega \| \geq \left| \langle \Omega, \upt - \Omega \rangle \right|  \geq c_0  \alpha^{-2} |t| - C \alpha^{-4} t^2 .
\end{align}
Hence, for $t = \delta \alpha^2$ and $\delta > 0$ small enough, there exists a constant $C_\delta > 0$ such that $\| \upt - \Omega \| \geq C_\delta$, uniformly in $\alpha$.

\section*{Acknowledgements}

Financial support by the European Union's Horizon 2020 research and innovation programme under  the Marie Sk\l{}odowska-Curie grant agreement No. 754411 (S.R.) and the European Research Council under grant agreement  No. 694227 (N.L. and R.S.), as well as by the SNSF Eccellenza project PCEFP2 181153 (N.L.), the NCCR SwissMAP (N.L. and B.S.) and by the \emph{Deutsche  Forschungsgemeinschaft (DFG)} through the Research Training Group 1838: \emph{Spectral Theory and Dynamics of Quantum Systems} (D.M.) is gratefully acknowledged. 
B.S. gratefully acknowledges financial support from the Swiss National Science Foundation 
through the Grant ``Dynamical and energetic properties of Bose-Einstein condensates'' and from the European Research 
Council through the ERC-AdG CLaQS (grant agreement No 834782).  
D.M. thanks Marcel Griesemer for helpful discussions.

\end{document}